%% file: hedonic.tex
\documentclass[letterpaper,USenglish,cleveref, autoref, thm-restate]{lipics-v2021}

\title{Hedonic Games and Treewidth Revisited} 

\titlerunning{Hedonic Games and Treewidth Revisited} 

\author{Tesshu Hanaka}{Department of Mathematical Informatics, Graduate School of Informatics, Nagoya University}{hanaka@i.nagoya-u.ac.jp}{https://orcid.org/0000-0001-6943-856X}{}

\author{Michael Lampis}
{Université Paris-Dauphine, PSL University, CNRS,
LAMSADE, 75016, Paris,
France}{michail.lampis@lamsade.dauphine.fr}{https://orcid.org/0000-0002-5791-0887}{}

\authorrunning{Tesshu Hanaka and Michael Lampis} 

\Copyright{Tesshu Hanaka and Michael Lampis} 



\category{} 

\relatedversion{} 

\funding{This work is partially supported by PRC CNRS JSPS project PARAGA (Parameterized Approximation Graph Algorithms) JPJSBP 120192912 and by JSPS KAKENHI Grant Number JP19K21537, JP21K17707, 21H05852.}

\nolinenumbers 

\usepackage{tikz}

\EventEditors{John Q. Open and Joan R. Access}
\EventNoEds{2}
\EventLongTitle{42nd Conference on Very Important Topics (CVIT 2016)}
\EventShortTitle{CVIT 2016}
\EventAcronym{CVIT}
\EventYear{2016}
\EventDate{December 24--27, 2016}
\EventLocation{Little Whinging, United Kingdom}
\EventLogo{}
\SeriesVolume{42}
\ArticleNo{23}

\hideLIPIcs 

\begin{document}

\maketitle

\begin{abstract}

We revisit the complexity of the well-studied notion of Additively Separable
Hedonic Games (ASHGs). Such games model a basic clustering or coalition
formation scenario in which selfish agents are represented by the vertices of
an edge-weighted digraph $G=(V,E)$, and the weight of an arc $uv$ denotes the
utility $u$ gains by being in the same coalition as $v$.  We focus on
(arguably) the most basic stability question about such a game: given a graph,
does a Nash stable solution exist and can we find it efficiently?

We study the (parameterized) complexity of ASHG stability when the
underlying graph has treewidth $t$ and maximum degree $\Delta$. The current
best FPT algorithm for this case was claimed by Peters [AAAI 2016], with time
complexity roughly $2^{O(\Delta^5t)}$. We present an algorithm with parameter
dependence $(\Delta t)^{O(\Delta t)}$, significantly improving upon the
parameter dependence on $\Delta$ given by Peters, albeit with a slightly worse
dependence on $t$. Our main result is that this slight performance
deterioration with respect to $t$ is actually completely justified: we observe
that the previously claimed algorithm is incorrect, and that in fact no
algorithm can achieve dependence $t^{o(t)}$ for bounded-degree graphs, unless
the ETH fails.  This, together with corresponding bounds we provide on the
dependence on $\Delta$ and the joint parameter establishes that our algorithm
is essentially optimal for both parameters, under the ETH.


We then revisit the parameterization by treewidth alone and resolve a question
also posed by Peters by showing that Nash Stability remains strongly NP-hard on
stars under additive preferences.  Nevertheless, we also discover an island of
mild tractability: we show that Connected Nash Stability is solvable in
pseudo-polynomial time for constant $t$, though with an XP dependence on $t$
which, as we establish, cannot be avoided.

\end{abstract}

\section{Introduction}

Coalition formation is a topic of central importance in computational social
choice and in the mathematical social sciences in general. The goal of its
study is to understand how groups of selfish agents are likely to partition
themselves into teams or clusters, depending on their preferences. The most
well-studied case of coalition formation are \emph{hedonic games}, which have
the distinguishing characteristic that each agent's utility only depends on the
coalition on which she is placed (and not on the coalitions of other players).
Hedonic games have recently been an object of intense study also from the
computer science perspective
\cite{AloisioFV20,AzizBBHOP19,BarrotOSY19,BarrotY19,BoehmerE20,0001BW21,BullingerK21,0001MM21,IgarashiOSY19,OhtaBISY17,SliwinskiZ17},
due in part to their numerous applications in, among others, social network
analysis \cite{Olsen09}, scheduling group activities \cite{DarmannEKLSW18}, and
allocating tasks to wireless agents \cite{SaadHBDH11}. For more information we
refer the reader to \cite{Cechlarova16} and the relevant chapters of standard
computational social choice texbooks \cite{AzizS16}.


Hedonic games are extremely general and capture many interesting scenarios in
algorithmic game theory and computational social choice. Unfortunately, this
generality implies that most interesting questions about such games are
computationally hard; indeed, even encoding the preferences of agents generally
takes exponential space. This has motivated the study of natural succinctly
representable versions of hedonic games.  In this paper, we focus on one of the
most widely-studied such models called Additively-Separable Hedonic Games
(ASHG).  In this setting the interactions between agents are given by an
edge-weighted directed graph $G=(V,E)$, where the weight of an arc $uv\in E$
denotes the utility that $u$ gains by being placed in the same coalition as
$v$. Thus, vertices which are not connected by an arc are considered to be
indifferent to each other. Given a partition into coalitions, the utility of a
player $v$ is defined as the sum of the weights of out-going arcs from $v$ to
its own coalition. 

\begin{table}
\begin{tabular}{|l|l|l|} 
\hline
Parameter & Algorithms & Lower Bounds \\
\hline
$t,p$ &  & Strongly NP-hard for Stars (G) (\cref{thm:paraNP})\\
                    & $(nW)^{O(t^2)}$ (C) (\cref{thm:algtw})& No $f(p)\cdot n^{o(p/\log p)}$ (C) (\cref{thm:whard})\\
\hline
$t,p+\Delta$ & $\left(\Delta
t\right)^{O(\Delta t)} (n+\log W)^{O(1)}$ &  No $(p\Delta)^{o(p\Delta)}(nW)^{O(1)}$ (G) (\cref{thm:eth1})\\
 & (G)  (\cref{thm:algtwd})&  \\
 & & No $\Delta^{o(\Delta)}(nW)^{O(1)}$ if $p=O(1)$  (G) (\cref{cor:eth1})\\
 & & No $p^{o(p)}n^{O(1)}$ if $\Delta,W=O(1)$ (\cref{thm:eth2}) (G,C) \\
\hline
\end{tabular} \caption{Summary of results. $t,p,\Delta,W$ denote the treewidth,
pathwidth, maximum degree, and maximum absolute weight. Results denoted by (G)
apply to general (possibly disconnected) \textsc{Nash Stability}, and by (C) to
\textsc{Connected Nash Stability}.}\label{tbl:summary} \end{table}

A rich literature exists studying various questions about ASHGs, including a
large spectrum of stability concepts and social welfare maximization
\cite{AzizBS13,Ballester04,ElkindFF20,FlamminiKMZ21,HanakaKMO19,Olsen09,OlsenBT12,SungD10}.
In this paper we focus on perhaps the most basic notion of stability one may
consider.  We say that a configuration $\pi$ is \emph{Nash Stable} if no agent
$v$ can unilaterally strictly increase her utility by selecting a different
coalition of $\pi$ or by forming a singleton coalition.  The algorithmic
question that we are interested in studying is the following: given an ASHG,
does a Nash Stable partition exist? Even though other notions of stability
exist (notably when deviating players are allowed to collaborate
\cite{BranzeiL09,DeinekoW13,Peters17,Woeginger13}), fully understanding the
complexity of \textsc{Nash Stability} is of particular importance, because of
the fundamental nature of this notion.

\textsc{Nash Stability} of ASHGs has been thoroughly studied and is,
unfortunately, NP-complete. We therefore adopt a parameterized point of view
and investigate whether some desirable structure of the input can render the
problem tractable.  We consider two of the most well-studied graph parameters:
the treewidth $t$ and the maximum degree $\Delta$ of the underlying graph. The
study of ASHGs in this light was previously taken up by Peters \cite{Peters16a}
and the goal of our paper is to improve and clarify the state of the art given
by this previous work.

\subparagraph*{Summary of Results} Our results can be divided into two parts
(see \cref{tbl:summary} for a summary).  In the first part of the paper we
parameterize the problem by $t+\Delta$, that is, we study its complexity for
graphs that have simultaneously low treewidth and low maximum degree. The study
of hedonic games on such graphs was initiated by Peters \cite{Peters16a}, who
already considered a wide variety of algorithmic questions on ASHGs for these
parameters and provided FPT algorithms using Courcelle's theorem. Due to the
importance of \textsc{Nash Stability}, more refined algorithmic arguments were
given in the same work, and it was claimed that \textsc{Connected Nash
Stability} (the variant of the problem where coalitions must be connected in
the underlying graph) and \textsc{Nash Stability} can be decided with parameter
dependence roughly $2^{\Delta^2t}$ and $2^{\Delta^5t}$, respectively (though as
we explain below, these claims were not completely justified). We thus revisit
the problem with the goal of determining the optimal parameter dependence for
\textsc{Nash Stability} in terms of $t$ and $\Delta$.  Our positive
contribution is an algorithm deciding \textsc{Nash Stability} in time
$\left(\Delta t\right)^{O(\Delta t)} (n+\log W)^{O(1)}$, where $W$ is the
maximum absolute weight, significantly improving the parameter dependence for
$\Delta$ (\cref{thm:algtwd}).  This is achieved by reformulating the problem as
a coloring problem with $t\Delta$ colors in a way that encodes the property
that two vertices belong in the same coalition and then using dynamic
programming to solve this problem.  Our main technical contribution is then to
establish that our algorithm is essentially optimal. To that end we first show
that if there exists an algorithm solving \textsc{Nash Stability} in time
$(p\Delta)^{o(p\Delta)}(nW)^{O(1)}$, where $p$ is the pathwidth of the
underlying graph, then the ETH is false (\cref{thm:eth1}).  Hence, it is not
possible to obtain a better parameter dependence, even if we accept a
pseudo-polynomial running time and a more restricted parameter.

If we were considering a parameterization with a single parameter, at this
point we would be essentially done, since we have an algorithm and a lower
bound that match.  However, the fact that $\Delta$ and $t$ are two a priori
independent variables significantly complicates the analysis because,
informally, the space of running time functions that depend on two variables is
not totally ordered.  To see what we mean by that, recall that \cite{Peters16a}
claimed an algorithm with complexity roughly $2^{\Delta^5t}$, while our
algorithm's complexity has the form $(\Delta t)^{\Delta t}$.  The two
algorithms are not directly comparable in performance: for some values of
$\Delta,t$ one is better and for some the other (though the range of parameters
where $2^{\Delta^5t}<(\Delta t)^{\Delta t}$ is quite limited). As a result,
even though \cref{thm:eth1} shows that no algorithm can beat the algorithm of
\cref{thm:algtwd} in all cases, it does not rule out the possibility that some
algorithm beats it in \emph{some} cases, for example when $\Delta$ is much
smaller than $t$, or vice-versa.  We therefore need to work harder to argue
that our algorithm is indeed optimal in essentially all cases.  In particular,
we show that even if pathwidth is constant the problem cannot be solved in
$\Delta^{o(\Delta)}(nW)^{O(1)}$ (\cref{cor:eth1}); and even if $\Delta$ and $W$
are constant, the problem cannot be solved in $p^{o(p)}n^{O(1)}$
(\cref{thm:eth2}).  Hence, we succeed in covering essentially all corner cases,
showing that our algorithm's slightly super-exponential dependence on \emph{the
product} of $\Delta$ and $t$ is truly optimal, and we cannot avoid the slightly
super-exponential on either parameter, even if we were to accept a much worse
dependence on the other.

An astute reader will have noticed a contradiction between our lower bounds and
the algorithms of \cite{Peters16a}. It is also worth noting that
\cref{thm:eth2} applies to both the connected and disconnected cases of the
problem, using an argument due to \cite{Peters16a}.  Hence, \cref{thm:eth2}
implies that, either the ETH is false, or \emph{neither} of the aforementioned
algorithms of \cite{Peters16a} can have the claimed performance, as executing
them on the instances produced by our reduction (which have $\Delta=O(1)$)
would give parameter dependence $2^{O(t)}$, which is ruled out by
\cref{thm:eth2}.  Indeed, in \cref{sec:twd} we explain in more detail that the
argumentation of \cite{Peters16a} lacks an ingredient (the partition of
vertices in each neighborhood into coalitions) which turns out to be necessary
to obtain a correct algorithm and also key in showing the lower bound. Hence,
the slightly super-exponential dependence on $t$ cannot be avoided (under the
ETH), and the dependence on $t$ promised in \cite{Peters16a} is impossible to
achieve: the best one can hope for is the slightly super-exponential dependence
on both $t$ and $\Delta$ given in \cref{thm:algtwd}.

In the second part of the paper, we consider \textsc{Nash Stability} on graphs
of low treewidth, without making any further assumptions (in particular, we
consider graphs of arbitrarily large degree).  This parameterization was
considered by Peters \cite{Peters16a} who showed that the problem is strongly
NP-hard on stars and thus motivated the use of the double parameter $t+\Delta$.
This would initially appear to settle the problem.  However, we revisit this
question and make two key observations: first, the reduction of
\cite{Peters16a} does not show hardness for additive games, but for a more
general version of the problem where preferences of players are not necessarily
additive but are described by a collection of boolean formulas (HC-nets
\cite{ElkindW09,IeongS05}).  
It was therefore explicitly posed as an open question whether \emph{additive}
games are also hard \cite{Peters16a}.  Second, in the reduction of
\cite{Peters16a} coalitions are disconnected.  As noted in
\cite{IgarashiE16,Peters16a}, there are situations where Nash Stable coalitions
make more sense if they are connected in the underlying graph. We therefore ask
whether \textsc{Connected Nash Stability}, where we impose a connectivity
condition on coalitions, is an easier problem.  

Our first contribution is to resolve the open question of \cite{Peters16a} by
showing that imposing either one of these two modifications does \emph{not}
render the problem tractable: \textsc{Nash Stability} of additive hedonic games
is still strongly NP-hard on stars (\cref{thm:paraNP}); and \textsc{Connected
Nash Stability} of hedonic games encoded by HC-nets is still NP-hard on stars
(\cref{thm:paraNP2}).  However, our reductions stubbornly refuse to work for
the natural combination of these conditions, namely, \textsc{Connected Nash
Stability} for additive hedonic games on stars.  Surprisingly, we discover that
this is with good reason: \textsc{Connected Nash Stability} turns out to be
solvable in pseudopolynomial time on graphs of bounded treewidth
(\cref{thm:algtw}).  More precisely, our algorithm, which uses standard dynamic
programming techniques but crucially relies on the connectedness of coalitions,
runs in ``pseudo-XP'' time, that is, in polynomial time when $t=O(1)$ and
weights are polynomially bounded. Completing our investigation we show that
this is essentially best possible: obtaining a pseudo-polynomial time algorithm
with FPT dependence on treewidth (or pathwidth) would contradict standard
assumptions (\cref{thm:whard}).  Hence, in this part we establish that there is
an overlooked case of ASHGs that does become somewhat tractable when we only
parameterize by treewidth, but this tractability is limited.

\subparagraph*{Related work} Deciding if an ASHG admits a partition that is
Nash Stable or has other desirable properties is NP-hard
\cite{AzizBS13,Ballester04,Olsen09,PetersE15,SungD10}.  Hardness remains even
in cases where a Nash Stable solution is guaranteed, such as symmetric
preferences, where the problem is PLS-complete \cite{GairingS19}, and
non-negative preferences, where it is NP-hard to find a non-trivial stable
partition \cite{OlsenBT12}. The problem generally remains hard when we impose
the requirement that coalitions must be connected \cite{BiloGM19,IgarashiE16}.

A related \textsc{Min Stable Cut} problem is studied in \cite{Lampis21}, where
we partition the vertices into two coalitions in a Nash Stable way.
Interestingly, the complexity of that problem turns out to be $2^{O(\Delta
t)}$, since each vertex has $2$ choices; this nicely contrasts with
\textsc{Nash Stability}, where vertices have more choices, and which is
slightly super-exponential parameterized by treewidth.  Similar slightly
super-exponential complexities have been observed with other problems involving
treewidth and partitioning vertices into sets
\cite{HarutyunyanLM21,LokshtanovMS18}.

\section{Preliminaries}

We use standard graph-theoretic notation and assume that the reader is familiar
with standard notions in parameterized complexity, including treewidth and
pathwidth \cite{CyganFKLMPPS15}. We mostly deal with directed graphs and denote
an arc from vertex $u$ to vertex $v$ as $uv$. When we talk about the degree or
the neighborhood of a vertex $v$, we refer to its degree and its neighborhood
in the underlying graph, that is, the graph obtained by forgetting the
directions of all arcs. Throughout the paper $\Delta(G)$ (or simply $\Delta$,
when $G$ is clear from the context) denotes the maximum degree of the
underlying graph of $G$. The Exponential Time Hypothesis (ETH) is the
assumption that there exists $c>1$ such that \textsc{3-SAT} on formulas with
$n$ variables does not admit a $c^n$ algorithm \cite{ImpagliazzoPZ01}. We will
mostly use a somewhat simpler to state (and weaker) form of this assumption
stating that \textsc{3-SAT} cannot be solved in time $2^{o(n)}$.

%

In this paper we will be mostly interested in \emph{Additively Separable
Hedonic Games} (ASHG). In an ASHG we are given a directed graph $G=(V,E)$ and a
weight function $w:V\times V\to\mathbb{Z}$ that encodes agents' preferences.
The function $w$ has the property that for all $u,v\in V$ such that $uv\not\in
E$ we have $w(u,v)=0$, that is, non-zero weights are only given to arcs. A
solution to an ASHG is a partition $\pi$ of $V$, where we refer to the sets of
$V$ as classes or, more simply, as coalitions. For each $v\in V$ and
$S\subseteq V$ the utility that $v$ derives from being placed in the coalition
$S$ is defined as $p_v(S) = \sum_{u\in S\setminus\{v\}} w(v,u)$.  A partition
$\pi$ is Nash Stable if we have the following: for each $v\in V$, if $v$
belongs in the class $S$ of $\pi$, we have $p_v(S) \ge 0$ and for each $S'\in
\pi$ we have $p_v(S)\ge p_v(S')$. In other words, no vertex can strictly
increase its utility by joining another coalition of $\pi$ or forming a
singleton coalition.  We also consider the notion of \emph{Connected Nash
Stable} partitions, which are Nash Stable partitions $\pi$ with the added
property that all classes of $\pi$ are connected in the underlying undirected
graph of $G$.

\section{Parameterization by Treewidth and Degree}\label{sec:twd}

In this section we revisit \textsc{Nash Stability} parameterized by $t+\Delta$,
which was previously studied in \cite{Peters16a}.  Our main positive result is
an algorithm given in \cref{sec:algtwd} solving the problem with dependence
$(t\Delta)^{O(t\Delta)}$. 

Our main technical contribution is then to show in \cref{sec:twdhard} that this
algorithm is essentially optimal, under the ETH.  As explained, we need several
different reductions to settle this problem in a satisfactory way. The main
reduction is given in \cref{thm:eth1} and uses the fact that a partition
restricted to the neighborhood of a vertex with degree $\Delta$ encodes roughly
$\Delta\log\Delta$ bits of information, because there are around
$\Delta^{\Delta}$ partitions of $\Delta$ elements into equivalence classes.
This key idea allows the first reduction to compress the treewidth more and
more as $\Delta$ increases. Hence, we can produce instances where both $t$ and
$\Delta$ are super-constant, but appropriately chosen to match our bound. In
this way, \cref{thm:eth1} rules out running times of the form, say
$(t\Delta)^{t+\Delta}$, as when $t,\Delta$ are both super-constant,
$t+\Delta=o(t\Delta)$. By modifying the parameters of \cref{thm:eth1} we then
obtain \cref{cor:eth1} from the same construction, which states that no
algorithm can have dependence $\Delta^{o(\Delta)}$, even on graphs of bounded
pathwidth.  On the other hand, this type of construction cannot show hardness
for instances of bounded degree, as when $\Delta=O(1)$, then
$\Delta^{\Delta}=O(1)$, so we cannot really compress the treewidth of the
produced instance. Hence, we use a different reduction in \cref{thm:eth2},
showing that the problem cannot be solved with dependence $p^{o(p)}$ on
instances of bounded degree. This reduction uses a super-constant number of
coalitions that ``run through'' the graph, and hence produces instances with
super-constant $t$.  The three complementary reductions together cover the
whole range of possibilities and indicate that there is not much room for
improvement in our algorithm.

It is worth discussing here that, assuming the ETH, \cref{thm:eth2} contradicts
the claimed algorithms of \cite{Peters16a}, which for $\Delta=O(1)$ would solve
\textsc{(Connected) Nash Stability} with dependence $2^{O(t)}$, while
\cref{thm:eth2} claims that the problem cannot be solved in time $2^{o(t\log t)}$.
Let us then briefly explain why the proof sketch for these algorithms in
\cite{Peters16a} is incomplete: the idea of the algorithms is to solve
\textsc{Connected Nash Stability}, and use the arcs of the instance to verify
connectivity. Hence, the DP algorithm will remember, in a ball of distance $2$
around each vertex, which arcs have both of their endpoints in the same
coalition. The claim is that this information allows us to infer the
coalitions. Though this is true if one is given this information for the whole
graph, it is not true locally around a vertex where we only have information
about other vertices which are close by.  In particular, it could be the case
that $u$ has neighbors $v_1,v_2$, which happen to be in the same coalition, but
such that the path proving that this coalition is connected goes through
vertices far from $u$. Because this cannot be verified locally, any DP
algorithm would need to store some connectivity information about the vertices
in a bag which, as implied by \cref{thm:eth2} inevitably leads to a dependence
of the form $t^t$. 

\subsection{Improved FPT Algorithm}\label{sec:algtwd}

In order to obtain our algorithm for \textsc{Nash Stability} we will need two
ingredients.  The first ingredient will be a reformulation of the problem as a
vertex coloring problem.  We use the following definition where, informally, a
vertex is stable if its outgoing weight to vertices of the same color cannot be
increased by changing its color.

\begin{definition} A Stable $k$-Coloring of an edge-weighted digraph $G$ is a
function $c:V\to [k]$ satisfying the following property: for each $v\in V$ we
have $\sum_{u\in c^{-1}(c(v))} w(v,u) \ge \max_{j\in [k+1]} \sum_{u\in
c^{-1}(j)} w(v,u)$.  \end{definition}

Note that in the definition above we take the maximum over $j\in [k+1]$ of the
total weight of $v$ towards color class $j$. Since $c$ is a function that uses
$k$ colors, we have $c^{-1}(k+1)=\emptyset$ and hence this ensures that the
total weight of $v$ towards its own color must always be non-negative in a
stable coloring. Also note that to calculate the total weight from $v$ to a
certain color class $j$, it suffices to consider the vertices of color $j$ that
belong in the out-neighborhood of $v$.

Our strategy will be to show that, for appropriately chosen $k$, deciding
whether a graph admits a stable $k$-Coloring is equivalent to deciding whether
a Nash Stable partition exists. Then, the second ingredient of our approach is
to use standard dynamic programming techniques to solve Stable $k$-Coloring on
graphs of bounded treewidth and maximum degree.

The key lemma for the first part is the following:

\begin{lemma}\label{lem:stable} Let $G=(V,E)$ be an edge-weighted digraph whose
underlying graph has maximum degree $\Delta$ and admits a tree decomposition
with maximum bag size $t$.  Then, $G$ has a Nash Stable partition if and only
if it admits a Stable $k$-Coloring for $k=t\cdot \Delta$.  \end{lemma}

\begin{proof} First, suppose that we have a Stable $k$-Coloring $c:V\to [k]$ of
the graph for some value $k$.  We obtain a Nash Stable partition of $V(G)$ by
turning each color class into a coalition. By the definition of Stable
$k$-Coloring, each vertex has at least as high utility in its own color class
(and hence its own coalition) as in any other, so this partition is stable.

%

For the converse direction, suppose that there exists a Nash Stable partition
$\pi$ of $G$. We will first attempt to color the coalitions of $\pi$ in a way
that any two coalitions which are at distance at most two receive distinct
colors, while using at most $t\cdot \Delta$ colors. In the remainder, when we
refer to the distance between two sets of vertices $S_1,S_2$, we mean
$\min_{u\in S_1, v\in S_2}d(u,v)$, where distances are calculated in the
underlying graph.

Consider the graph $G^2$ obtained from the underlying graph of $G$ by
connecting any two vertices which are at distance at most $2$ in the underlying
graph of $G$. We can construct a tree decomposition of $G^2$ where all bags
contain at most $t\cdot \Delta$ vertices by taking the assumed tree
decomposition of $G$ and adding to each bag the neighbors of all vertices
contained in that bag.  Furthermore, we can assume without loss of generality
that any equivalence class $C$ of the Nash Stable partition $\pi$ is connected
in $G^2$.  If not, that would mean that there exists a class $C$ that contains
a connected component $C'\subseteq C$ such that $C'$ is at distance at least
$3$ from $C\setminus C'$ in the underlying graph of $G$. In that case we could
partition $C$ into two classes $C', C\setminus C'$, without affecting the
stability of the partition.

Formally now the claim we wish to make is the following:

\begin{claim}\label{claim} There is a coloring $c$ of the equivalence classes
of $\pi$ with $k=t\cdot \Delta$ colors such that any two classes $C_1,C_2$ of
$\pi$ which are at distance at most two in the underlying graph of $G$ receive
distinct colors.  \end{claim}

\begin{claimproof}

We prove the claim by induction on the number of equivalence classes of $\pi$.
If there is only one class the claim is trivial.

Consider a rooted tree decomposition of $G^2$. For an equivalence class $C$ of
$\pi$ we say that the bag $B$ is the top bag for $C$ if $B$ contains a vertex
of $C$ and no bag that is closer to the root contains a vertex of $C$. Select
an equivalence class $C$ of $\pi$ whose top bag is as far from the root as
possible. We claim that there are at most $t\cdot\Delta-1$ classes $C'$ which
are at distance at most $2$ from $C$ in $G$.  

In order to prove that there are at most $t\cdot\Delta-1$ other classes at
distance at most two from $C$, consider such a class $C'$, which is therefore
at distance one from $C$ in $G^2$. Let $B$ be the top bag of $C$.  If $C'$ does
not contain any vertex that appears in $B$ then we get a contradiction as
follows: first, $C'$ has a neighbor of a vertex of $C$, so these two vertices
must appear together in a bag; since all vertices of $C$ appear in the sub-tree
rooted at $B$, some vertices of $C'$ must appear strictly below $B$ in the
decomposition; since $B$ is a separator of $G^2$ and $C'$ is connected, if no
vertex of $C'$ is in $B$ then all vertices of $C'$ appear below $B$ in the
decomposition; but then, this contradicts the choice of $C$ as the class whose
top bag is as far from the root as possible. As a result, for each $C'$ that is
a neighbor of $C$ in $G^2$, there exists a distinct vertex of $C'$ in $B$.
Since $|B|\le t\cdot \Delta$ and $B$ contains a vertex of $C$, we get that the
coalitions $C'$ which are neighbors of $C$ in $G^2$ are at most $t\cdot
\Delta-1$.  

We now remove all vertices of $C$ from the graph and claim that $\pi$
restricted to the new graph is still a Nash Stable partition. By induction,
there is a coloring of the remaining coalitions of $\pi$ that satisfies the
claim. We keep this coloring and assign to $C$ a color that is not used by any
of the at most $k-1$ coalitions which are at distance two from $C$. Hence, we
obtain the claimed coloring of the classes of $\pi$.  \end{claimproof}

From \cref{claim} we obtain a coloring of the equivalence classes of $\pi$ with
$k=t\cdot \Delta$ colors, such that any two equivalence classes which are at
distance at most $2$ in the underlying graph of $G$ receive distinct colors. We
now obtain a coloring of $V$ by assigning to each vertex the color of its
class. In the out-neighborhood of each vertex $v$ the partition induced by the
coloring is the same as that induced by $\pi$, since all the vertices in the
out-neighborhood of $v$ are at distance at most $2$ from each other in $G$.
Hence, the $k$-Coloring must be stable, because otherwise a vertex would have
incentive to deviate in $\pi$ by joining another coalition or by becoming a
singleton.  \end{proof}


\begin{theorem}\label{thm:algtwd} There exists an algorithm which, given an
ASHG defined on a digraph $G=(V,E)$ whose underlying graph has maximum degree
$\Delta$ and a tree decomposition of the underlying graph of $G$ of width $t$,
decides if a Nash Stable partition exists in time $\left(\Delta
t\right)^{O(\Delta t)} (n+\log W)^{O(1)}$, where $n=|V|$ and $W$ is the largest
absolute weight.\end{theorem}

\begin{proof}

Using \cref{lem:stable} we will formulate an algorithm that decides if the
given instance admits a Stable $k$-Coloring for $k=(t+1)\Delta$, since this is
equivalent to deciding if a Nash Stable partition exists. We first obtain a
tree decomposition of $G^2$ by placing into each bag of the given decomposition
all the neighbors of all the vertices of the bag. 

We now execute a standard dynamic programming algorithm for $k$-coloring on
this new decomposition, so we sketch the details. The DP table has size
$k^{(t+1)\Delta} = (\Delta t)^{O(\Delta t)}$ since we need to store as a
signature of a partial solution the colors of all vertices contained in a bag.
The only difference with the standard DP algorithm for coloring is that our
algorithm, whenever a new vertex $v$ is introduced in a bag $B$, considers all
possible colors for $v$, and then for each $u\in B$, if all neighbors of $u$
are contained in $B$, verifies for each signature whether $u$ is stable.
Signatures where a vertex is not stable are discarded. The key property is now
that for any vertex $u$, there exists a bag $B$ such that $B$ contains $u$ and
all its neighbors (since in $G^2$ the neighborhood of $u$ is a clique), hence
only signatures for which all vertices are stable may survive until the root of
the decomposition.  \end{proof}

\subsection{Tight ETH-based Lower Bounds}\label{sec:twdhard}

\begin{theorem}\label{thm:eth1} If the ETH is true, there is no algorithm which
decides if an ASHG on a graph with $n$ vertices, maximum degree $\Delta$, and
pathwidth $p$ admits a Nash Stable partition in time
$(p\Delta)^{o(p\Delta)}(nW)^{O(1)}$, where $W$ is the maximum absolute weight.
\end{theorem}

\begin{figure}
\input{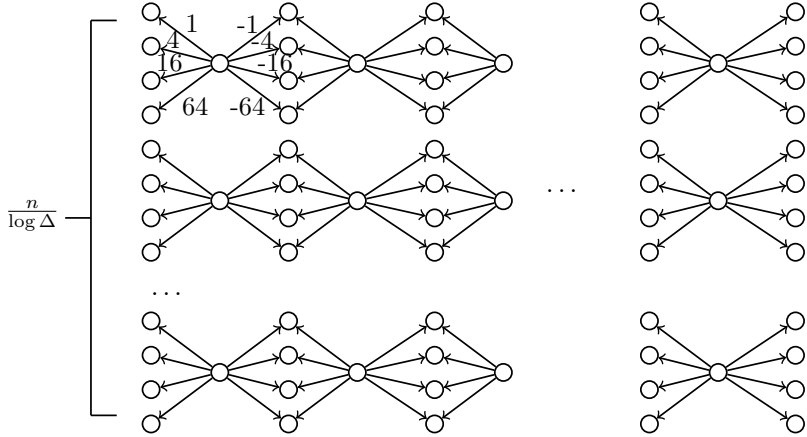} \caption{Overview of the reduction of \cref{thm:eth1}.
Selection vertices form $m$ columns of $n/\log\Delta$ vertices each. An
assignment is encoded by the partition of a column into coalitions. The
$\frac{n}{\Delta\log\Delta}$ consistency vertices that follow a column ensure
that the partition is repeated in the next column, because consistency vertices
are disliked by everyone, so the only way to make the coalition stable is to
make sure they have utility $0$ everywhere. }\label{fig1} \end{figure}

\begin{proof}

We will give a parametric reduction which, starting from a 3-\textsc{SAT}
instance $\phi$ with $n$ variables and $m$ clauses, and for any desired
parameter $\Delta<n/\log n$, constructs an ASHG instance $G$ with the following
properties:

\begin{enumerate} 

\item $G$ can be constructed in time polynomial in $n$ 

\item $G$ has maximum degree $O(\Delta)$ 

\item $G$ has pathwidth $O(\frac{n}{\Delta\log\Delta})$ 

\item the maximum absolute value $W$ is $2^{O(\Delta)}$ 

\item  $\phi$ is satisfiable if and only if there exists a Nash Stable
partition.

\end{enumerate}

Before we go on, let us argue why a reduction that satisfies these properties
does indeed establish the theorem: given a 3-\textsc{SAT} instance on $n$
variables, we set $\Delta=\lfloor\sqrt{n}\rfloor$.  We construct $G$ in
polynomial time, therefore the size of $G$ is polynomially bounded by $n$.
Deciding if $G$ has a Nash Stable partition is equivalent to solving $\phi$ by
the last property. By the third property, the pathwidth of the constructed
graph is $O(\frac{\sqrt{n}}{\log n})$, so $p\Delta = O(\frac{n}{\log n})$.
Furthermore, $W=2^{O(\sqrt{n})}$. If deciding if a Nash Stable partition exists
can be done in time $(p\Delta)^{o(p\Delta)}(|G|\cdot W)^{O(1)}$, the total
running time for deciding $\phi$ is $(p\Delta)^{o(p\Delta)} (|G|\cdot W)^{O(1)}
= 2^{o(n)}$ contradicting the ETH.

We now describe our construction. We are given a 3-\textsc{SAT} instance $\phi$
with variables $x_0,\ldots,x_{n-1}$, and a parameter $\Delta$, which we assume
to be a power of $2$ (otherwise we increase its value by at most a factor of
$2$).  We also assume without loss of generality that all clauses of $\phi$
have size exactly $3$ (otherwise we repeat literals). We construct the
following graph:

\begin{enumerate}

\item \textbf{Selection vertices}: for each $i_1\in \{0,\ldots,\lceil
\frac{n}{\Delta\log\Delta}\rceil\}$, $i_2\in\{0,\ldots,\Delta-1\}$,
$j\in\{1,\ldots,m\}$, we construct a vertex $u_{(i_1,i_2,j)}$.

\item \textbf{Consistency vertices}: for each
$i_1\in\{0,\ldots,\lceil\frac{n}{\Delta\log\Delta}\rceil\}$,
$j\in\{1,\ldots,m-1\}$, we construct a vertex $c_{(i_1,j)}$. For
$i_2\in\{0,\ldots,\Delta-1\}$ we give weights: $w(c_{(i_1,j)},u_{(i_1,i_2,j)})
= 4^{i_2}$; $w(c_{(i_1,j)},u_{(i_1,i_2,j+1)}) = -4^{i_2}$;
$w(u_{(i_1,i_2,j)},c_{(i_1,j)}) = w(u_{(i_1,i_2,j+1)},c_{(i_1,j)}) =
-4^{\Delta}$.

\item \textbf{Clause gadget}: for each $j\in\{1,\ldots,m\}$ we construct two
vertices $s_j, s_j'$ and set $w(s_j,s_j')=2$. We also construct three vertices
$\ell_{(j,1)}, \ell_{(j,2)}, \ell_{(j,3)}$ and set $w(\ell_{(j,1)},s_j) =
w(\ell_{(j,2)},s_j)=w(\ell_{(j,3)},s_j)=2$ and
$w(s_j,\ell_{(j,1)})=w(s_j,\ell_{(j,2)})=w(s_j,\ell_{(j,3)})=-1$. 

\item \textbf{Palette gadget}: we construct a vertex $p$ and a helper $p'$. We
set $w(p,p')=w(p',p)=1$. Furthermore, for
$i_1=\lceil\frac{n}{\Delta\log\Delta}\rceil$ and for all $i_2\in
\{0,\ldots,\Delta-1\}$, we set $w(p,u_{(i_1,i_2,0)})=1$ and
$w(u_{(i_1,i_2,0)},p)=-1$.

\end{enumerate}

So far, we have described the main part of our construction, without yet
specifying how we encode which literals appear in each clause. Before we move
on to describe this part, let us give some intuition about the construction up
to this point. The intended meaning of the palette gadget is that vertices
$u_{(i_1,i_2,0)}$ for $i_1=\lceil\frac{n}{\Delta\log\Delta}\rceil$ and $i_2\in
\{0,\ldots,\Delta-1\}$ should be placed in distinct coalitions ($p$ can be
thought of as a stalker).  These vertices form a ``palette'', in the sense that
every other selection vertex encodes an assignment to some of the variables of
$\phi$ by deciding which of the palette vertices it will join. Hence, we intend
to extract an assignment of $\phi$ from a stable partition by considering each
vertex $u_{(i_1,i_2,0)}$, for $i_1\in \{0,\ldots, \lceil
\frac{n}{\Delta\log\Delta}\rceil-1\}$, $i_2\in \{0,\ldots,\Delta-1\}$.  For
each such vertex we test in which of the $\Delta$ palette partitions the vertex
was placed, and this gives us enough information to encode $\log\Delta$
variables of $\phi$.  Since we have $\lceil\frac{n}{\Delta\log\Delta}\rceil
\cdot \Delta \ge \frac{n}{\log\Delta}$ non-palette selection vertices, and each
such selection vertex encodes $\log\Delta$ variables, we will be able to encode
an assignment to $n$ variables. The role of the consistency vertices is to make
sure that the partition of the selection vertices (and hence, the encoded
assignment) stays consistent throughout our construction.

In order to complete the construction, let us make the above intuition more
formal. For $i_1\in \{0,\ldots,\lceil\frac{n}{\Delta\log \Delta}\rceil-1\}$,
$i_2\in \{0,\ldots,\Delta-1\}$ and for any $j\in\{1,\ldots,m\}$, we will say
that $u_{(i_1,i_2,j)}$ encodes the assignment to variables $x_k$, with $k\in
\{i_1\cdot \Delta\log\Delta +i_2\log\Delta, \cdots, i_1\cdot \Delta\log\Delta
+i_2\log\Delta+\log\Delta-1\}$. Equivalently, given an integer $k$, we can
compute which selection vertices encode the assignment to $x_k$ by setting
$i_1=\lfloor \frac{k}{\Delta\log\Delta}\rfloor$ and $i_2 = \lfloor
\frac{k-i_1\Delta\log \Delta}{\log\Delta} \rfloor$. In that case, $x_k$ is
represented by $u_{(i_1,i_2,j)}$ (for any $j$).

Let us now explain precisely how an assignment to the variables of $\phi$ is
encoded by the placement of selection vertices in coalitions. Let $k$ be such
that $x_k$ is encoded by $u_{(i_1,i_2,j)}$ and let $i_3 = k-
i_1\Delta\log\Delta -i_2\log\Delta$. We have $i_3\in\{0,\ldots,\log\Delta-1\}$.
If $x_k$ is set to True in the assignment, then $u_{(i_1,i_2,j)}$ must be
placed in the same coalition as a palette vertex $u_{\lceil
\frac{n}{\Delta\log\Delta}\rceil, i_2',0}$ where $i_2'$ has the following
property: if we write $i_2'$ in binary, then the bit in position $i_3$ must be
set to $1$. Similarly, if $x_k$ is set to False, then we must place
$u_{(i_1,i_2,j)}$ in the same coalition as a palette vertex $u_{\lceil
\frac{n}{\Delta\log\Delta}\rceil, i_2',0}$ where writing $i_2'$ in binary gives
a $0$ in position $i_3$. Observe that, given an assignment and a vertex
$u_{(i_1,i_2,j)}$ which represents $\log\Delta$ variables, this process fully
specifies the palette vertex with which we must place $u_{(i_1,i_2,j)}$ to
represent the assignment.  In the converse direction, we can extract from the
placement of $u_{(i_1,i_2,j)}$ an assignment to the vertices it represents if
we know that all palette vertices are placed in distinct components, simply by
finding the palette vertex
$u_{(\lceil\frac{n}{\Delta\log\Delta}\rceil,i_2',0)}$ in the coalition of
$u_{(i_1,i_2)}$, writing down $i_2'$ in binary, and using its $\log\Delta$ bits
in order to give an assignment to the $\log\Delta$ variables represented by
$u_{(i_1,i_2,j)}$.

We are now ready to complete the construction by considering each clause. Each
vertex $\ell_{(j,\alpha)}$, $\alpha\in\{1,2,3\}$, corresponds to a literal of
the $j$-th clause of $\phi$.  If this literal involves the variable $x_k$, we
calculate integers $i_1,i_2,i_3$ from $k$ as explained in the previous
paragraph. Say, $x_k$ is the $i_3$-th variable represented by
$u_{(i_1,i_2,j)}$.  We set $w(\ell_{(j,\alpha)},u_{(i_1,i_2,j)})=1$.
Furthermore, for each $i_2'\in\{0,\ldots,\Delta-1\}$ we look at the $i_3$-th
bit of the binary representation of $i_2'$. If setting $x_k$ to the value of
that bit would make the literal represented by $\ell_{(j,\alpha)}$ True, we set
$w(\ell_{(j,\alpha)},u_{(\lceil\frac{n}{\Delta\log\Delta}\rceil,i_2',j)})=1$;
otherwise we set
$w(\ell_{(j,\alpha)},u_{(\lceil\frac{n}{\Delta\log\Delta}\rceil,i_2',j)})=0$.
We perform the above process for all $j\in\{1,\ldots,m\}$,
$\alpha\in\{1,2,3\}$.

Our construction is now complete, so we need to show that we satisfy all the
claimed properties. It is not hard to see that the graph can be built in
polynomial time, and the maximum absolute weight used is $2^{O(\Delta)}$ (on
arcs incident on some consistency vertices). The vertices with maximum degree
are the consistency vertices and the vertices representing literals, both of
which have degree $O(\Delta)$.

To establish the bound on the pathwidth we first delete $p,p'$ from the graph,
as this can decrease pathwidth by at most $2$. Now observe that, for each $j$,
the set $C_j= \{ c_{(i_1,j)}\ |\
i_1\in\{0,\ldots,\lceil\frac{n}{\Delta\log\Delta}\rceil\}\ \}$ is a separator
of the graph. We claim that if we fix a $j$, then the set $C_j\cup C_{j+1}$
separates the set $C'_j = \{ u_{(i_1,i_2,j)}\ |
i_1\in\{0,\ldots,\lceil\frac{n}{\Delta\log\Delta}\rceil\},
i_2\in\{0,\ldots,\Delta-1\}\ \} \cup \{s_j, s_j', \ell_{(j,1)}, \ell_{(j,2)},
\ell_{(j,3)}\}$ from the rest of the graph. We claim that we can calculate a
path decomposition of the graph induced by $C_j\cup C'_j\cup C_{j+1}$ with
width $O(\frac{n}{\Delta\log\Delta})$ such that the first bag contains $C_j$
and the last bag contains $C_{j+1}$. If we achieve this we can construct a path
decomposition of the whole graph by gluing these decompositions together in the
obvious way (in order of increasing $j$). However, a path decomposition of this
induced subgraph can be constructed by placing $C_j\cup C_{j+1}\cup \{s_j,
s_j', \ell_{(j,1)}, \ell_{(j,2)}, \ell_{(j,3)}\}$ and a distinct vertex of the
remainder of $C'_j$ in each bag. This decomposition has width $2|C_j|+O(1) =
O(\frac{n}{\Delta\log\Delta})$.

Finally, let us establish the main property of the construction, namely that
$\phi$ is satisfiable if and only if the ASHG instance admits a Nash Stable
partition. If there exists a satisfying assignment to $\phi$ we construct a
partition as follows: (i) $p,p'$ are in their own coalition (ii) each
consistency vertex is a singleton (iii) for $i_2\in \{0,\ldots,\Delta-1\}$, the
vertices of $\{ u_{\lceil\frac{n}{\Delta\log\Delta}\rceil,i_2,j}\ |\
j\in\{1,\ldots,m\}\}$ are placed in a distinct coalition (iv) we place the
remaining selection vertices in one of the previous $\Delta$ coalitions in a
way that represents the assignment as previously explained (v) for each
$j\in\{1,\ldots,m\}$ the $j$-th clause contains a True literal; we place the
corresponding vertex $\ell_{(j,\alpha)}$ together with its out-neighbor in the
selection vertices, and the remaining literal vertices together with $s,s'$ in
a new coalition. We claim that this partition is Nash Stable. We have the
following argument: (i) $p'$ is with $p$, while $p$ cannot increase her utility
by leaving $p'$, since all its other out-neighbors are in distinct coalitions
(ii) for each $i_1,i_2,j$, the vertices $u_{(i_1,i_2,j)}, u_{(i_1,i_2,j+1)}$
are in the same coalition. Hence, the utility of each consistency vertex is $0$
in any coalition, and such vertices are stable as singletons (iii) each
selection vertex $u_{(i_1,i_2,j)}$ has utility $0$, and such vertices only have
out-going arcs of negative weight (iv) in each clause gadget we have a
coalition with $s_j, s_j'$ together with two literal vertices, say
$\ell_{(j,1)}, \ell_{(j,2)}$; no vertex has incentive to leave this coalition
(v) finally, for literal vertices $\ell_{(j,\alpha)}$ which we placed together
with a selection vertex, we observe that if the assignment sets the
corresponding literal to True, the selection vertex that is an out-neighbor of
$\ell_{(j,\alpha)}$ must have been placed in a coalition that contains a
palette vertex towards which $\ell_{(j,\alpha)}$ has positive utility, hence
the utility of $\ell_{(j,\alpha)}$ is $2$ and this vertex is stable.

For the converse direction, suppose that we have a Nash Stable partition $\pi$.
We first prove that all vertices
$u_{\lceil\frac{n}{\Delta\log\Delta}\rceil,i_2,0}$, for
$i_2\in\{0,\ldots,\Delta-1\}$, must be in distinct coalitions. Indeed, if two
of them are in the same coalition, $p$ will have incentive to join the
coalition that has the maximum number of such vertices. However, once $p$ joins
such a coalition, these vertices will have negative utility, contradicting
stability. Second, we prove that for each $i_1,i_2,j$, the vertices
$u_{(i_1,i_2,j)}, u_{(i_1,i_2,j+1)}$ must be in the same coalition. If not,
consider two such vertices which are in distinct coalitions and maximize $i_2$.
We claim that in this case $c_{(i_1,j)}$ will always join $u_{(i_1,i_2,j)}$.
Indeed, from the selection of $i_2$, we have that for $i_2'>i_2$, the
contribution of arcs with absolute weight $4^{i_2'}$ to the utility of
$c_{(i_1,j)}$ cancels out; while for $i_2'<i_2$ the sum of all absolute
utilities of arcs with weights $4^{i_2'}$ is too low to affect the placement of
$c_{(i_1,j)}$ (in particular, $4^{i_2}-\sum_{j<i_2} 4^j > \sum_{j<i_2} 4^j$).
But, if $c_{(i_1,j)}$ joins such a coalition, a selection vertex has negative
utility, contradicting stability.

From the two properties above we can now extract an assignment to $\phi$. For
each selection vertex $u_{(i_1,i_2,j)}$, if this vertex is in the same
coalition as $u_{(\lceil\frac{n}{\Delta\log\Delta}\rceil,i_2',0)}$, we give an
assignment to the variables represented by $u_{(i_1,i_2,j)}$ as described, that
is, we write $i_2'$ in binary and use one bit for each variable. Note that the
choice of $j$ here is irrelevant, as we have shown that thanks to the
consistency vertices, for each $i_1,i_2$, all vertices $u_{(i_1,i_2,j)}$ are in
the same coalition. If $u_{(i_1,i_2,j)}$ is not in the same coalition as any
$u_{(\lceil\frac{n}{\Delta\log\Delta}\rceil,i_2',0)}$, we set its corresponding
variables in an arbitrary way. To see that this assignment satisfies clause
$j$, consider $s_j$, which, without loss of generality is placed with $s_j'$.
If three of the vertices $\ell_{(j,1)}, \ell_{(j,2)}, \ell_{(j,3)}$ are in the
same coalition as $s_j$, then $s_j$ has negative utility, contradiction. Hence,
one of these vertices, say $\ell_{(j,1)}$, is in another coalition. But then,
since the neighbors of this vertex among vertices
$u_{(\lceil\frac{n}{\Delta\log\Delta}\rceil,i_2,j)}$ are all in distinct
coalitions, $\ell_{(j,1)}$ is in the same coalition with one such vertex and
its out-neighbor selection vertex. But this means that we have extracted an
assignment from the corresponding vertex and that this assignment sets the
corresponding literal to True, satisfying the clause.  \end{proof}

\begin{corollary}\label{cor:eth1} If the ETH is true, there is no algorithm
which decides if an ASHG on a graph with $n$ vertices, maximum degree $\Delta$,
and constant pathwidth admits a Nash Stable partition in time
$\Delta^{o(\Delta)}(nW)^{O(1)}$, where $W$ is the maximum absolute weight.
\end{corollary}

\begin{proof}

We use the same reduction as in \cref{thm:eth1}, from a \textsc{3-SAT} formula
on $n$ variables, but set $\Delta = \lfloor \frac{n}{2\log n}\rfloor$.
According to the properties of the construction, the pathwidth of the resulting
graph is $O(\frac{n}{\Delta\log\Delta}) = O(1)$, the maximum degree is
$O(n/\log n)$, the maximum weight is $2^{O(n/\log n)}$ and the size of the
constructed graph is polynomial in $n$. If there exists an algorithm for
finding a Nash Stable partition in the stated time, this gives a $2^{o(n)}$
algorithm for \textsc{3-SAT}.  \end{proof}

\begin{theorem}\label{thm:eth2} If the ETH is true, there is no algorithm which
decides if an ASHG on a graph with $n$ vertices, constant maximum degree
$\Delta$, and pathwidth $p$ admits a Nash Stable partition in time
$p^{o(p)}n^{O(1)}$, even if all weights have absolute value $O(1)$.
\end{theorem}

\begin{figure}[h] \begin{tabular}{l|l} \input{redPw1.tex} & \input{redPw2.tex}
\end{tabular} \caption{On the left, an overview of the reduction of
\cref{thm:eth2}. We have $n+m$ columns, each with $\sqrt{n}$ palette vertices
and $\frac{2n}{\log n}$ selection vertices. Assignments are encoded by the
placement of selection vertices in coalitions.  On the right, an OR gadget,
where the right-most part depicts the checker vertices. Such a vertex is
satisfied if its two out-going arcs going to the rest of the graph lead to the
same coalition.  Otherwise, the checker joins the Or gadget. On the left, the
vertices of the Or gadget starting from $r_1$ at the top. Each $r_i$ has
utility $2$ for $r_{i+1}$ but utility $-1$ for $r_{i-1}$. Each $r_i$ has two
vertices attached, one that it likes ($r_i'$) and one that it dislikes
($r_i''$). If the checker attached to $r_{k_0}$ joins the rest of the graph, we
place $r_{k_0}''$ with $r_{k_0+1}$ and continue in this way to obtain a stable
partition of the Or gadget.}\label{fig2} \end{figure}
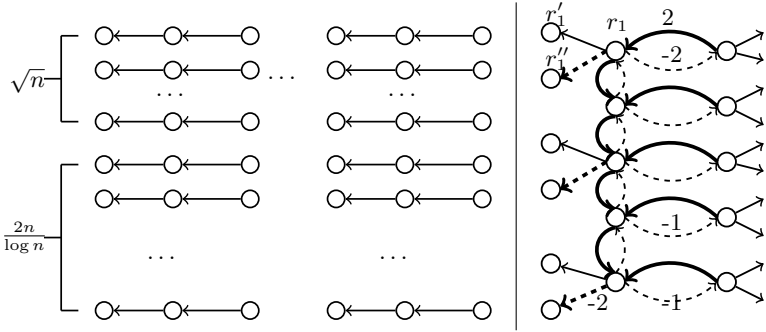

\begin{proof}

We describe a reduction from a \textsc{3-SAT} formula $\phi$ with $n$ variables
and $m$ clauses. Our goal is to build an equivalent instance with bounded
maximum degree, bounded  maximum weight, and pathwidth $O(n/\log n)$. Suppose
without loss of generality that $n$ is a power of $4$ (otherwise add some dummy
variables), and the variables of $\phi$ are $x_0,x_1,\ldots, x_{n-1}$. We
construct a graph initially made up of the following parts:

\begin{enumerate}

\item \textbf{Palette Paths}: For $i\in\{0,\ldots,\sqrt{n}-1\}$,
$j\in\{1,\ldots,m+n\}$, we construct a vertex $p_{(i,j)}$. For
$j\in\{1,\ldots,m+n-1\}$ we set $w(p_{(i,j+1)}, p_{(i,j)}) = 1$.

\item \textbf{Selection Paths}: For $i\in \{0,\ldots, \lfloor\frac{2n}{\log
n}\rfloor\}$, $j\in \{1,\ldots,m+n\}$, we construct a vertex $u_{(i,j)}$.  For
$i\in \{0,\ldots, \lfloor\frac{2n}{\log n}\rfloor\}$, $j\in \{1,\ldots,m+n-1\}$
we set $w(u_{(i,j+1)}, u_{(i,j)})=1$.

\item \textbf{Palette Consistency Gadget}: For each pair of indices $i,i' \in
\{1,\ldots,\sqrt{n}\}$, with $i\neq i'$, we arbitrarily select a distinct index
$j\in \{m+1,\ldots,m+n-1\}$. We construct two vertices $a_j, b_j$ and set
$w(a_j, p_{(i,j)}) = 1$, $w(a_j, p_{(i',j)}) = -1$, $w(a_j,b_j)=1$,
$w(b_j,a_j)=w(b_j,p_{(i,j)})=w(b_j,p_{(i',j)})=-1$.

\end{enumerate}

At this point we have described the skeleton of our construction which will be
sufficient to encode the variables of the original formula and their
assignments.  Before we proceed to explain how we complete the construction to
encode the clauses, we give some intuition. The $\sqrt{n}$ palette paths and
the roughly $2n/\log n$ selection paths are intended to form coalitions, in the
sense that for a fixed $i$, all vertices $p_{(i,j)}$ must belong in the same
coalition, and similarly for all vertices of $u_{(i,j)}$. To ensure this, we
will make sure that vertices $p_{(i,j)}, u_{(i,j)}$ have no other out-going
arcs in our construction, hence each such vertex will always have an incentive
to join its immediate neighbor in the path. The palette consistency gadgets
will make sure that the $\sqrt{n}$ palette paths form $\sqrt{n}$ distinct
coalitions.

Armed with this intuition, we now explain how assignments will be encoded.
Assuming the $\sqrt{n}$ palette paths form distinct coalitions, we can decide
to place $u_{(i,1)}$ (and its corresponding selection path) inside any one of
these $\sqrt{n}$ coalitions. This choice encodes $\log(\sqrt{n})=\frac{\log
n}{2}$ bits of information (which is an integer, because $n$ is a power of
$4$). Hence, we define that the placement of $u_{(i,1)}$ encodes the assignment
of variables $x_k$ for $k\in\{\frac{i\log n}{2},\ldots, \frac{(i+1)\log
n}{2}-1\}$. Equivalently, given $k$, we say that the assignment of $x_k$ is
encoded by the placement of the vertex $u_{(i,1)}$, where $i=\lfloor
\frac{2k}{\log n}\rfloor$. To be more precise we will make the following
correspondence: the placement of $u_{(i,1)}$ dictates that $x_k$ is set to True
if $i=\lfloor \frac{2k}{\log n}\rfloor$, $u_{(i,1)}$ is in the same coalition
as $p_{(i',1)}$, and the binary representation of $i'$ using $\frac{\log n}{2}$
bits has a $1$ at position $k-\frac{i\log n}{2}$ (where we number positions in
the binary representation starting from $0$); otherwise the placement of
$u_{(i,1)}$ dictates that $x_k$ is set to False. It is easy to also make this
correspondence in the opposite direction: if we have an assignment to the
variables represented by $u_{(i,1)}$, we write these variables in binary in
order of increasing index and let $i'$ be the resulting number. We place
$u_{(i,1)}$ together with $p_{(i',1)}$.

Now that we have explained our intended encoding of the variable assignments we
can complete the construction. Fix a $j\in\{1,\ldots,m\}$ and consider the
$j$-th clause of $\phi$ which, without loss of generality, contains three
literals (if not, we can repeat literals). Suppose the three (not necessarily
distinct) variables involved in the clause are $x_{k_1}, x_{k_2}, x_{k_3}$, and
$i_1=\lfloor \frac{2k_1}{\log n}\rfloor$ (and $i_2,i_3$ are defined similarly).
We construct the following gadgets:

\begin{enumerate}

\item \textbf{Indegree reduction}: Construct three directed paths of length
$\sqrt{n}$.  Label their vertices $\ell_{(j,\alpha,\beta)}$, for
$\alpha\in\{1,2,3\}$ and $\beta\in\{0,\ldots,\sqrt{n}-1\}$. For all
$\alpha\in\{1,2,3\}$ and $\beta\in\{0,\ldots,\sqrt{n}-2\}$ we set
$w(\ell_{(j,\alpha,\beta)}, \ell_{(j,\alpha,\beta+1)})=1$. We also set
$w(\ell_{(j,\alpha,\sqrt{n})},u_{(i_{\alpha},j)})=1$. 

\item \textbf{Checker vertices}: For each $\alpha\in\{1,2,3\}$ we do the
following: for each $i'\in\{0,\ldots,\sqrt{n}-1\}$ we consider whether the
assignment encoded by placing $u_{(i_{\alpha},j)}$ in the coalition of
$p_{(i',j)}$ would satisfy the literal involving $x_{k_{\alpha}}$ (i.e.
whether the binary representation of $i'$ has a $1$ at position
$k_{\alpha}-\frac{i_{\alpha}\log n}{2}$ if the literal is positive, and $0$ if
the literal is negated).  If yes, we construct a checker vertex
$c_{(j,\alpha,i')}$ and set $w(c_{(j,\alpha,i')},p_{(i',j)}) =
w(c_{(j,\alpha,i')},\ell_{(j,\alpha,i')}) =1$. Let $C_j$ be the set containing
all checker vertices we constructed in this step for a given $j$ (for all
$\alpha\in\{1,2,3\}$ and $i'\in\{0,\ldots,\sqrt{n}-1\}$). We have $|C_j|\le
3\sqrt{n}$.

\item \textbf{Or gadget}: We construct for each $k\in\{1,\ldots,|C_j|\}$ three
vertices $r_k, r_k', r_k''$ and set for all $k$, $w(r_k, r_k')=1$, and
$w(r_k,r_k'')=-2$. Furthermore, for all $k\in\{1,\ldots, |C_j|-1\}$ we set
$w(r_k, r_{k+1}) = 2$ and $w(r_{k+1},r_k)=-1$. For each
$k\in\{2,\ldots,|C_j|\}$ we pick a distinct vertex $c\in C_j$ and set
$w(p_k,c)=-1$ and $w(c,p_k)=2$. Finally, for the remaining vertex $c$ of $C_j$
we set $w(r_1,c)=-2$ and $w(c,r_1)=2$.

\end{enumerate}

The construction described above is repeated for each $j\in\{1,\ldots,m\}$, in
order to encode all $m$ clauses of the instance. Let us give some intuition:
first, the indegree reduction paths are not particularly important; all
vertices $\ell_{(j,\alpha,\beta)}$ are intended to belong in the coalition of
$u_{(i_{\alpha},j)}$, and their role is only to allow us to avoid giving this
vertex large in-degree (we re-route arcs that would have gone to
$u_{(i_{\alpha},j)}$ towards distinct vertices of the path). The checker
vertices play the following role: if the encoded assignment sets a literal to
True, then one of the checkers will have utility $2$ by joining the coalition
of a vertex $u_{(i_{\alpha},j)}$. In this case we say that this checker is
``satisfied''. Other checkers will join the coalition of their out-neighbor in
the Or gadget. Hence, the role of the Or gadget is to make sure that at least
one checker vertex must be satisfied to obtain a stable partition.

Let us now prove that our construction has all the necessary properties. First,
it is not hard to see that the maximum degree $\Delta$ and maximum absolute
weight $W$ are bounded by a constant. We claim that the pathwidth of our
construction is $O(n/\log n)$. To see this, let $B_j=\{ u_{(i,j)}\ |\
i\in\{0,\ldots,\lfloor\frac{2n}{\log n}\rfloor \}\} \cup \{ p_{(i,j)}\ |\
i\in\{0,\ldots,\sqrt{n}-1\} \}$. We construct a path decomposition using
$n+m-1$ bags, where for $j\in\{1,\ldots,n+m-1\}$, the $j$-th bag contains
$B_j\cup B_{j+1}$. This decomposition has width $O(n/\log n)$ and already
covers all palette and selection vertices and their induced edges. To complete
the decomposition, for each $j\in\{1,\ldots,m\}$, we add to the $j$-th bag all
the (at most $O(\sqrt{n})$) vertices we constructed to represent clause $j$
(that is, the Or gadget, checkers, and indegree reduction vertices for clause
$j$). Furthermore, for $j\in\{m+1,\ldots,m+n-1\}$, we add to the $j$-th bag the
palette consistency vertices $a_j,b_j$, if they exist. We obtain a
decomposition of width $O(n/\log n)$. Hence, if we prove that the new instance
has a Nash Stable partition if and only if $\phi$ is satisfiable, we are done.
Indeed, in that case an algorithm with running time $p^{o(p)}n^{O(1)}$ would
run in $(n/\log n)^{o(n/\log n)} = 2^{o(n)}$ and would refute the ETH.

What remains then is to prove that $\phi$ is satisfiable if and only if the
ASHG instance we constructed has a stable partition. For the forward direction,
suppose there exists a satisfying assignment. We construct a stable partition
as follows: initially, for each $i\in\{0,\ldots,\sqrt{n}-1\}$, each palette
path $P_i = \{ p_{(i,j)}\ |\ j\in\{1,\ldots,m+n\}\}$ forms its own coalition;
furthermore for each $i\in\{0,\ldots,\lfloor\frac{2n}{\log n}\rfloor\}$, all
vertices of the set $\{ u_{(i,j)}\ |\ j\in\{1,\ldots,m+n\}\}$ are placed in
$P_{i'}$, where $i'$ is obtained by writing the assignments to the variables
$x_k$ for $k\in\{\frac{i\log n}{2}, \frac{i\log n}{2}+1,\ldots,\frac{(i+1)\log
n}{2}-1\}$ and reading it as a binary number. Observe that all vertices
described so far are stable. For palette consistency vertices $a_j,b_j$, we
place $b_j$ as a singleton (which is stable), and $a_j$ together with its
out-neighbor in the palette vertices that gives it positive utility. This is
always possible, since each $P_i$ is in a distinct coalition. For the clause
gadgets, fix a $j$, and place all vertices $\ell_{(j,\alpha,\beta)}$ in the
same coalition as $u_{(i_{\alpha},j)}$. This is stable for these vertices (and
indifferent for $u_{(i_{\alpha},j)}$). Because we have a satisfying assignment,
there is a literal that is set to True, say the literal involving variable
$x_{k_\alpha}$. This implies that there exists $i'$ and checker vertex
$c_{(j,\alpha,i')}$ such that the checker has positive utility for $p_{(i',j)}$
and $\ell_{(j,\alpha,i')}$, and the latter two vertices are in the same
coalition. We place the checker in this coalition, where it receives utility
$2$ and is therefore stable. For each other checker  $c\in C_j$, we place $c$
together with its out-neighbor in the Or gadget, making $c$ stable. Finally,
there exists a $k_0\in \{1,\ldots,|C_j|\}$ such that the neighbor of $r_{k_0}$
in $C_j$ is not placed together with $r_{k_0}$. We place vertices of the Or
gadget in coalitions as follows: for $k\in \{1,\ldots,k_0-1\}$ we place
$r_k,r_k'$ together with $r_{k+1}$, and $r_k''$ as a singleton; for
$k\in\{k_0,\ldots,|C_j|-1\}$ we place $r_k$ together with $r_k'$ and place
$r_k''$ together with $r_{k+1}$; finally, $r_{|C_j|}$ is placed with $r_k'$.
This partition is stable because for $k<k_0$ the vertex $r_k$ receives utility
$2$ from its arc towards $r_{k+1}$ and $1$ from $r_k'$; $r_{k_0}$ receives at
most $-1$ from $r_{k_0-1}$ (if $k_0>1$) but also $1$ from $r_{k_0}'$, so its
utility is not negative; furthermore, since $r_{k_0}'',r_{k_0+1}$ are together
$r_{k_0}$ cannot increase its utility by switching; the same arguments apply
for $|C_j|>k>k_0$ while for $r_{|C_j|}$ its utility is also non-negative and
this vertex is stable.

For the converse direction, suppose that there exists a stable partition $\pi$.
We first observe that for all $i\in\{0,\ldots,\sqrt{n}-1\}$, $P_i$ is contained
in a coalition, otherwise, there would be a $p_{(i,j+1)}$ in a coalition
distinct from that of $p_{(i,j)}$, but then the former vertex would have
incentive to deviate. Furthermore, for $i\neq i'$, $P_i, P_{i'}$ are contained
in distinct coalitions. To see this, consider the palette consistency gadget
$a_j,b_j$ we constructed for the pair $i,i'$. The vertex $b_j$ has to be a
singleton (placing it together with one of its neighbors gives it negative
utility). Therefore, $a_j$ must receive positive utility in another coalition.
However, this would be impossible if the neighbors of $a_{j}$ in $P_i, P_{i'}$
were in the same coalition. We also observe that, for
$i\in\{0,\ldots,\lfloor\frac{2n}{\log n}\rfloor\}$ the vertices of the $i$-th
selection path belong in the same coalition (with arguments similar to those
for $P_i$). Hence, from this placement we extract an assignment for $\phi$. If
the vertex $u_{(i,1)}$ is placed together with $p_{(i',1)}$, we write $i'$ in
binary and use the bits to give values to the variables $x_k$ for $k\in\{
\frac{i\log n}{2},\ldots,\frac{(i+1)\log n}{2}-1\}$. If $u_{(i,1)}$ is not
together with any palette vertex, we set these variables arbitrarily.

We claim that the assignment we have extracted satisfies $\phi$. To see this,
consider the $j$-th clause. By arguments similar as above, all vertices of the
path $\ell_{(j,\alpha,\beta)}$ are placed together with $u_{(i_{\alpha},j)}$,
because each such vertex only has one out-going arc, and this arc has positive
weight.  We observe that if one of the checker vertices of $c_j$ is satisfied,
that is, if $c_j$ is placed in a coalition that does not contain its neighbor
in the Or gadget, the utility of $c_j$ in its current coalition must be $2$,
because checker vertices only have three out-going arcs, one with weight $2$
(towards the Or gadget) and two with weight $1$. Hence, $c_j$ must be placed in
the same component as a vertex $u_{(i_{\alpha},j)}$ and a palette vertex
$p_{(i',j)}$, and furthermore, the placement of $u_{(i_{\alpha},j)}$ in the
coalition of $P_{i'}$ encodes an assignment that satisfies the clause
(otherwise this checker would not have been constructed).  We conclude that if
there exists a $c_j$ that is not placed together with its neighbor in the Or
gadget, the clause is satisfied.  What remains, then, is to show that if each
checker vertex was placed together with its neighbor in the Or gadget, the
partition $\pi$ would be unstable.  Indeed, we observe that in this case $r_1$
must be placed with $r_2$ (otherwise $r_1$ has negative utility). But we also
note that if $r_k$ is placed together with $r_{k+1}$, then $r_{k+1}$ must be
placed together with $r_{k+2}$ (otherwise $r_{k+1}$ has negative utility).
Hence, all vertices $r_k$ for $k\in\{1,\ldots,|C_j|\}$ must be in the same
coalition. But then, the utility of $r_{|C_j|}$ is negative, contradiction.
\end{proof}

\begin{corollary}\label{cor:eth2} \cref{thm:eth2} also applies to
\textsc{Connected Nash Stability}.  \end{corollary}

\begin{proof}

We use an argument observed by Peters \cite{Peters16a} to reduce the problem of
finding a (possibly disconnected) Nash Stable partition, to the problem of
finding a connected Nash Stable partition. Consider an ASHG instance $G$ with
maximum degree $\Delta=O(1)$, maximum absolute weight $W=O(1)$ and pathwidth
$p$.  According to \cref{thm:eth2}, it is impossible to decide if $G$ admits a
Nash Stable partition in time $p^{o(p)}n^{O(1)}$. We construct a new instance
$G^2$ by adding an arc of weight $0$ between any two vertices of $G$ which are
at distance exactly two in the underlying graph. We claim that $G^2$ has (i)
bounded maximum degree, as the maximum degree is now $\Delta^2$ (ii) pathwidth
$O(p)$, or more precisely, pathwidth upper-bounded by $p\Delta$, since we can
obtain a decomposition of $G^2$ by taking a decomposition of $G$ and adding to
each bag the neighbors of all its vertices. Finally, $G^2$ has a connected Nash
Stable partition if and only if $G$ has a Nash Stable partition. One direction
is trivial, since we did not change the preferences of any agent. For the other
direction, if $G$ has a (possibly disconnected) Nash Stable partition $\pi$, we
check if $\pi$ (which is stable in $G^2$) becomes connected in $G^2$. If yes,
we are done. If not, this means there exists $C\in \pi$ such that $C$ contains
a component $C_1\subseteq C$ which is at distance at least $3$ from all
vertices of $C\setminus C_1$ in the underlying graph of $G$. But then, we can
obtain a new stable partition of $G$ by splitting $C$ into $C_1$ and
$C\setminus C_1$. This does not change the utility of any agent, and it also
does not create a new option for any agent, as anyone who has an arc towards
$C$, either has arcs towards $C_1$ or towards $C\setminus C_1$. We continue in
this way until $\pi$ is connected in $G^2$. We conclude that if there was an
algorithm with parameter dependence $p^{o(p)}$ for connected Nash Stability on
bounded degree graphs, we would obtain such an algorithm for general Nash
Stability on bounded degree graphs, contradicting the ETH. \end{proof}

\section{Parameterization by Treewidth Only}\label{sec:tw}

In this section we consider \textsc{Nash Stability} on graphs of bounded
treewidth. Peters \cite{Peters16a} showed that this problem is strongly NP-hard
on stars, but for a more general version where preferences are described by
boolean formulas (HC-nets). In \cref{sec:paraNP} we strengthen this hardness
result by showing that \textsc{Nash Stability} remains strongly NP-hard on
stars for additive preferences. We also show that \textsc{Connected Nash
Stability} is strongly NP-hard on stars, albeit also using HC-nets.

The only case that remains is \textsc{Connected Nash Stability} with additive
preferences. Somewhat surprisingly, we show that this case evades our hardness
results because it \emph{is} in fact more tractable.  We establish this via an
algorithm running in pseudo-polynomial time when the treewidth is constant in
\cref{sec:pseudoXP}. As a result, this is the only case of the problem which is
not strongly NP-hard on bounded treewidth graphs (unless P=NP).

We then observe that our algorithm only establishes that the problem is in XP
parameterized by treewidth (for weights written in unary). We show in
\cref{sec:binPack} that this is inevitable, as the problem is W[1]-hard
parameterized by treewidth even when weights are constant.  Hence, our
``pseudo-XP'' algorithm  is qualitatively optimal.

\subsection{Refined paraNP-hardnesss}\label{sec:paraNP}

\begin{theorem}\label{thm:paraNP} \textsc{Nash Stability} is strongly NP-hard
for stars for additive preferences.  \end{theorem}

\begin{proof} We present a reduction from \textsc{3-Partition}. In this problem
we are given a set of $3n$ positive integers $A$, a target value $T$, and are
asked to partition $A$ into $n$ triples, such that each triple has sum exactly
$T$. This problem has long been known to be strongly NP-hard \cite{GareyJ79}.
Furthermore, we can assume that the sum of all elements of $A$ is $nT$
(otherwise the answer is clearly No); and that all elements have values
strictly between $T/4$ and $T/2$, so sets of sizes other than three cannot have
sum $T$ (this can be achieved by adding $T$ to all elements and setting $4T$ as
the new target). 

We construct an ASHG as follows: for each element of $A$ we construct a vertex;
we construct a set $B$ of $n$ additional vertices; we add a ``stalker'' vertex
$s$ and a helper $s'$. The preferences are defined as follows: for all $x\in
A\cup B$ we set $w(x,s)=-1$; for each $x\in B$ we set $w(s,x)=2T$; for each
$x\in A$ we set $w(s,x)=-w(x)$, where $w(x)$ is the value of the corresponding
element in the original instance. Finally, we set $w(s,s')=T$ and $w(s',s)=1$.
The graph is a star as all arcs are incident on $s$.

If there exists a valid 3-partition of $A$, we construct a stable partition of
the new instance by placing $s$ with $s'$ and, for each triple placing its
elements in a coalition with a distinct vertex of $B$. Vertices of $A\cup B$
have utility $0$ in this configuration and no incentive to deviate; while $s$
would have utility $T$ in any existing coalition, so it has no incentive to
leave $s'$; $s'$ is satisfied as she is together with $s$.

For the converse direction, if we have a stable configuration $\pi$, $s'$ must
be with $s$ (otherwise $s'$ has incentive to deviate). Furthermore, $s$ cannot
be with any vertex of $A\cup B$, as placing $s$ with any such vertex would give
that vertex incentive to leave. Hence, $s,s'$ are one coalition of the stable
partition, and $s$ has utility $T$ in this coalition. This implies that every
coalition formed by vertices of $A\cup B$ must have utility at most $T$ for
$s$.

We now want to prove that every coalition of vertices of $A\cup B$ contains
exactly one vertex of $B$. If we show this, then the weight of elements of $A$
placed in each such coalition must be at least $T$, hence it must be exactly
$T$ (as the sum of all elements of $A$ is $nT$). Therefore, we obtain a
solution to the original instance.

To prove that every coalition that contains vertices of $A\cup B$ must contain
exactly one vertex of $B$, suppose first the there exists a coalition that only
contains vertices of $A$. Call the union of all such coalitions $A'\subseteq
A$.  Let $C_1,\ldots, C_k$ be the coalitions that contain some vertex of $B$,
for some $k\le |B|=n$.  We now reach a contradiction as follows: first, since
$s$ does not have incentive to join $C_i$, for $i\in[k]$, we have $\sum_{v\in
C_i} w(s,v) \le T$, therefore $\sum_{i=1}^k\sum_{v\in C_i} w(s,v) \le kT \le
nT$. On the other hand, $\sum_{i=1}^k\sum_{v\in C_i} w(s,v) \ge \sum_{v\in B}
w(s,v) + \sum_{v\in A\setminus A'} w(s,v) > 2nT - nT = nT$, because if $A'$ is
non-empty $\sum_{v\in A\setminus A'} w(s,v) < nT$. Hence we have a
contradiction and from now on we suppose that every coalition that contains a
vertex of $A\cup B$ has non-empty intersection with $B$.

Finally, consider a coalition that contains $k\ge 1$ vertices of $B$.  These
vertices give $s$ utility $2kT$, meaning that the sum of weights of vertices of
$A$ placed in this coalition must be at least $(2k-1)T$. Let $t_i$ be the
number of coalitions which contain exactly $i\ge 1$ vertices of $B$. We obtain
the inequality $\sum_i t_i (2i-1)T \le nT$, because the weight of all elements
of $A$ is $nT$. On the other hand $\sum_i it_i = n$, as we have that $|B|=n$.
We therefore have $\sum_i t_i(2i-1) \le n \Leftrightarrow \sum_i t_i \ge n =
\sum_i it_i \Leftrightarrow \sum_{i>1} (1-i)t_i \ge 0$, which can only hold if
$t_i=0$ for $i>1$.  \end{proof}

\begin{theorem}\label{thm:paraNP2} Deciding if a graphical hedonic game
represented by an HC-net admits a connected Nash Stable partition is NP-hard
even if the input graph is a star and all weights are in $\{1,-1\}$.
\end{theorem}

\begin{proof}

We present a reduction from \textsc{3-SAT}. Before we proceed, let us briefly
explain that in hedonic games representable by HC-nets, the utility of a vertex
$u$ in a coalition $S$ is calculated as a function of $N(u)\cap S$, using a set
of given ``rules''. A rule is a disjunctive term stating that some vertices of
$N(u)$ must or must not be present in $S$ to activate the rule. Each activated
rule has a pre-defined pay-off and the utility of $u$ is the sum of pay-offs of
activated rules.

Given a CNF formula $\phi$ with $n$ variables and $m$ clauses, we construct a
central vertex $s$, $2n$ literal vertices $x_1, \bar{x}_1, x_2,
\bar{x}_2,\ldots, x_n, \bar{x}_n$, and $m$ clause vertices $c_1,\ldots,c_m$.
The vertices form a star with $s$ as center. For every $c_j$ we define its
utility to be $1$ if it is together with $s$. For $s$ we have the following
rules: for each $i\in\{1,\ldots,n\}$, $s$ has utility $-1$ if both $x_i,
\bar{x}_i$ are in its coalition; for each clause $c_j$, $s$ has utility $-1$ if
$c_j$ is in its coalition; for each clause $c_j$ and each of the (at most 7)
assignments to its literals that satisfy the clause, we add a rule saying that
$s$ has utility $1$ if the literals of this assignment are all in its coalition
and their negations are not in the coalition.

Suppose $\phi$ is satisfiable: we form one coalition with $s$, all clause
vertices $c_j$, and all true literals of a satisfying assignment; all other
literal vertices are singletons. This partition is connected and stable. In
particular, $s$ has utility $0$ (it receives $-1$ from each clause vertex, but
$+1$ from satisfying each clause) and all $c_j$ have utility $1$.  For the
converse direction, in a stable partition $s$ is in the same coalition as at
most one of $x_i, \neg x_i$, for all $i\in\{1,\ldots,n\}$, otherwise it has
negative utility, which means it prefers to be alone. From this we can extract
an assignment to $\phi$. This assignment must satisfy all clauses because all
$c_j$ are with $s$ (giving it utility $-m$), so $m$ rules giving it utility $1$
must be activated, and for each clause at most one such rule can be activated.
\end{proof}

\subsection{Pseudo-XP algorithm for Connected Partitions}\label{sec:pseudoXP}

\begin{theorem}\label{thm:algtw}There exists an algorithm which, given an ASHG instance on $n$
vertices with maximum absolute weight $W$, along with a tree decomposition of
the underlying graph of width $t$, decides if a connected Nash Stable partition
exists in time $(nW)^{O(t^2)}$. \end{theorem}

\begin{proof}[\cref{thm:algtw}]

Our algorithm performs dynamic programming on the tree decomposition following
standard techniques, so we sketch some of the details and focus on the
non-trivial parts of the algorithm. As usual, we assume we have a nice tree
decomposition \cite{CyganFKLMPPS15} and the main challenge is in defining a
notion of signature of a solution, that is, the information that will be stored
in each bag of the decomposition that will allow us to encode the structure of
a solution as it interacts with the bag.

Consider a rooted nice tree decomposition, a bag $B$ and let $B^{\downarrow}$
be the set that contains all vertices of the input graph $G$ that appear in $B$
or in a descendant of $B$. The signature of a partition $\pi$ of $G=(V,E)$ with
respect to $B$ is a collection of the following information:

\begin{enumerate}

\item A partition $\pi_1$ of $B$ into equivalence classes, such that $x,y\in B$
are in the same class of $\pi_1$ if and only if $x,y$ are in the same coalition
of $\pi$ (so $\pi_1$ is the restriction of $\pi$ to $B$).

\item A partition $\pi_2$ of $B$ into equivalence classes, such that $x,y\in B$
are in the same class of $\pi_2$ if and only if $x,y$ are in the same coalition
of $\pi$ and there exists a path in the underlying graph of $G[B^{\downarrow}]$
whose internal vertices are in the same coalition of $\pi$ as $x,y$. Observe
that $\pi_2$ is necessarily a refinement of $\pi_1$. Informally, since $\pi$ is
a connected Nash Stable partition, the classes of $\pi_1$ must eventually
induce connected subgraphs. The partition $\pi_2$ tells which parts of each
class are already connected in $B^{\downarrow}$.

\item For each $x\in B$ its utility to its own coalition, that is, the sum of
the weights of arcs $(x,y)$ where $y\in B^{\downarrow}$ and $y$ is in the same
class of $\pi$ as $x$.

\item For each $x,y\in B$, such that $x,y$ are not in the same class of
$\pi_1$, the utility that $x$ would have if she joined $y$'s coalition, that
is, the sum of the weights of arcs $(x,y')$, where $y'\in B^{\downarrow}$ and
$y'$ is in the same class of $\pi$ as $y$.

\item For each $x\in B$ its maximum utility to any coalition that contains a
neighbor of $x$ and whose vertices are contained in $B^{\downarrow}\setminus
B$, that is, for each such equivalence class $C$ of $\pi$ that is fully
contained in $B^{\downarrow}\setminus B$ we compute $\sum_{y\in C} w(x,y)$ and
store the maximum of these values in the signature.

\end{enumerate} 

Informally, for each $x\in B$ we store, in addition to its placement with
respect to the other vertices of $B$, the utility that this vertex has in its
current coalition, the utility that it would have if it joined the coalition of
another vertex of $B$, and the utility that it would obtain if it joined the
best (in its view) coalition that only contains vertices that appear strictly
lower in the tree decomposition. We note here that a key observation is that
the coalitions which contain a vertex of $B^{\downarrow}\setminus B$ but no
vertex of $B$ are already complete, in the sense that such a coalition cannot
contain a vertex of $V\setminus B^{\downarrow}$ (in that case it would become
disconnected). This ensures that the utility that $x$ would have by joining
such a coalition cannot change as we move up the tree decomposition and
consider more vertices of $V\setminus B^{\downarrow}$. Intuitively, this is the
key property that explains why looking for connected Nash Stable partitions has
lower complexity than looking for (possibly disconnected) Nash Stable
partitions.

Having described the information that we store in our DP table, the rest of the
algorithm only needs to ensure that we appropriately update our tables for
Introduce, Join, and Forget nodes.  Introducing a vertex $x$ is
straightforward, as we consider all signatures contained in the child bag and
for each such signature we consider all the ways we could insert the new vertex
in $\pi_1, \pi_2$ and update weights according to the weights of arcs incident
on $x$. If $x$ creates a path between two vertices of its class of $\pi_1$
which are in distinct classes of $\pi_2$, we merge the two classes of $\pi_2$.
Crucially, $x$ has no neighbors in $B^{\downarrow}\setminus B$, so its utility
to all coalitions contained in this set is $0$. 

Forgetting a vertex is also straightforward, except that we need to make sure
that, according to the current signature the vertex is stable in its coalition
and its coalition is connected. Hence, when forgetting $x\in B$ we discard
all signatures where $x$ has strictly higher utility in a coalition other than
its own and all signatures where $x$ has negative utility in its own coalition;
furthermore we discard solutions where $x$ is the only vertex of its class in
$\pi_2$ and there exists a $y\in B$ such that $x,y$ are in the same class of
$\pi_1$ but in distinct classes of $\pi_2$.  (Informally, $\pi_1$ is the
partition into connected coalitions we intend to form, and $\pi_2$ is the
connectivity we have already assured, so if $x$ is not yet in the same
component as some other vertex $y$ in its coalition, the coalition will end up
being disconnected, with $x,y$ in distinct components).  When forgetting $x$,
if the class of $x$ in $\pi_1$ was a singleton, we also update the weights of
each remaining $y\in B$ by taking into account that the coalition that contains
$x$ is now contained in $B^{\downarrow}\setminus B$ (so we compare the utility
that $y$ would obtain by joining with the maximum utility it has in any such
coalition and update the maximum accordingly).

Finally, for Join nodes, we only consider pairs of signatures from the children
bag that agree on $\pi_1$. We combine the two partitions for $\pi_2$ in the
straightforward way to obtain a transitive closure. Finally, we update the
utility that each $x\in B$ has to the coalition of each $y\in B$ by adding the
utilities it has in the two sub-trees (taking care not to double count the arcs
contained in $B$).

The algorithm we sketched runs in time polynomial in the size of the DP tables,
so what remains is to bound the number of possible signatures. The number of
partitions of each bag is $t^{O(t)}$, while the utility of a vertex in any
coalition is always in $[-nW,nW]$, as the maximum absolute weight is $W$. For
each pair $x\in B$ we store $t+1$ such utilities in the worst case, so there
are at most $(nW)^{O(t^2)}$ possible distinct signatures.  \end{proof}


\subsection{W-hardness for Connected Partitions}\label{sec:binPack}

\begin{theorem}\label{thm:whard} If the ETH is true, deciding if an ASHG of
pathwidth $p$ admits a connected Nash Stable configuration cannot be done in
time $f(p)\cdot n^{o(p/\log p)}$ for any computable function $f$, even if all
weights are in $\{-1,1\}$.  \end{theorem}

\begin{proof}

We present a reduction from \textsc{Bin Packing}.  It was shown in
\cite{JansenKMS13} that \textsc{Bin Packing} with $n$ items and $k$ bins cannot
be solved in time $f(k)\cdot n^{o(k/\log k)}$, assuming the ETH, even if
weights are given in unary (that is, weights are polynomially bounded in $n$).
Recall that in an instance of $k$-\textsc{Bin Packing} we are given $n$
positive integers (the items) and a bin capacity $B>0$ and our goal is to
partition the $n$ items into $k$ sets such that each set has total sum at most
$B$. We can assume without loss of generality that the sum of the integers
given is exactly $kB$ (if the sum is strictly higher the answer is clearly No,
while if the sum is strictly lower we can pad the instance with items of weight
$1$). 

We construct an ASHG as follows: we construct $k$ vertices $b_1,\ldots,b_k$
representing the bins; we construct $k$ helpers $b_1',\ldots, b_k'$ and set for
each $i$ weight $w(b_i,b_i')=B$; we construct a vertex $v_i$ for each item and
set $w(v_i,b_j) = 1$ for all $j\in\{1,\ldots,k\}$ and $w(b_j,v_i)= -w(v_i)$ for
all $j$, where $w(v_i)$ is the weight of this item in the \textsc{Bin Packing}
instance.

If the \textsc{Bin Packing} instance admits a solution, we form $k$ coalitions
by placing in the $i$-th coalition the vertices $b_i,b_i'$ and all items placed
in bin $i$. We observe that this partition is stable, because vertices
representing items have utility $1$ and cannot increase their utility by
changing sets; vertices $b_i$ have utility $0$ and cannot obtain positive
utility by abandoning $b_i'$; and vertices $b_i'$ are indifferent. 

Conversely, if the ASHG has a connected Nash Stable configuration, we can see
that no coalition may contain vertices $v_i$ representing items of total weight
more than $B$. To see this, observe that such a coalition must contain a vertex
$b_i$ (otherwise it would be disconnected), but then that vertex will have
negative utility. Furthermore, no $v_i$ can be alone, since these vertices
always have an incentive to join some other vertex.  Hence, a Nash Stable
partition gives a partitition of the items into at most $k$ groups of weight
$B$.

The graph constructed has vertex cover $k$, hence also treewidth and pathwidth
$\le k$.  To complete the proof we observe that an edge $e=(u,v)$ of weight
$w(u,v)$ can be replaced by introducing $w(u,v)$ new vertices,
$e_1,\ldots,e_{w(u,v)}$ and setting $w(e_i,v)=1$ and
$w(u,e_i)=\textrm{sgn}(w(u,v))$, where $\textrm{sgn}(x)$ is $1$ if $x$ is
positive and $-1$ otherwise. Without loss of generality $e_i$ is always in the
same coalition as $v$ in any connected Nash Stable partition, so the solution
is preserved. Furthermore, it is not hard to see that this modification does
not increase the pathwidth of the graph. \end{proof}

By a slight modification of the previous proof we also obtain weak NP-hardness
for the case where the input graph has vertex cover $2$. 

\begin{corollary}\label{cor:weak} It is weakly NP-hard to decide if an ASHG on
a graph with vertex cover $2$ admits a connected Nash Stable partition.
\end{corollary}

\begin{proof}

We perform the same reduction as in \cref{thm:whard}, except we start from an
instance of $2$-\textsc{Bin Packing}, which is also known as \textsc{Partition}
and we do not perform the last step to obtain edges with weights in $\{-1,1\}$.
\textsc{Partition} is only weakly NP-hard \cite{GareyJ79}, so we obtain weak
NP-hardness. We note that a very similar reduction was given in
\cite{HanakaKMO19}, but for the problem where preferences are symmetric and we
seek to find a stable partition of maximum social utility.  \end{proof}

\section{Conclusions and Open Problems}

Our results give strong evidence that the precise complexity of \textsc{Nash
Stability} parameterized by $t+\Delta$ is in the order of
$(t\Delta)^{O(t\Delta)}$. It would be interesting to verify if the same is true
for \textsc{Connected Nash Stability}, as this problem turned out to be
slightly easier when parameterized only by treewidth, and is only covered by
\cref{cor:eth2} for the case of bounded-degree graphs. Of course, it would also
be worthwhile to investigate the fine-grained complexity of other notions of
stability. In particular, versions which are complete for higher levels of the
polynomial hierarchy \cite{Peters17} may well turn out to have
double-exponential (or worse) complexity parameterized by treewidth
\cite{LampisMM18,LampisM17}.

\bibliography{hedonic}



\end{document}

%% file: redPw1.tex
\ifx\du\undefined
  \newlength{\du}
\fi
\setlength{\du}{6.5\unitlength}
\begin{tikzpicture}[even odd rule]
\pgftransformxscale{1.000000}
\pgftransformyscale{-1.000000}
\definecolor{dialinecolor}{rgb}{0.000000, 0.000000, 0.000000}
\pgfsetstrokecolor{dialinecolor}
\pgfsetstrokeopacity{1.000000}
\definecolor{diafillcolor}{rgb}{1.000000, 1.000000, 1.000000}
\pgfsetfillcolor{diafillcolor}
\pgfsetfillopacity{1.000000}
\pgfsetlinewidth{0.100000\du}
\pgfsetdash{}{0pt}
\definecolor{diafillcolor}{rgb}{1.000000, 1.000000, 1.000000}
\pgfsetfillcolor{diafillcolor}
\pgfsetfillopacity{1.000000}
\pgfpathellipse{\pgfpoint{10.500000\du}{9.000000\du}}{\pgfpoint{0.500000\du}{0\du}}{\pgfpoint{0\du}{0.500000\du}}
\pgfusepath{fill}
\definecolor{dialinecolor}{rgb}{0.000000, 0.000000, 0.000000}
\pgfsetstrokecolor{dialinecolor}
\pgfsetstrokeopacity{1.000000}
\pgfpathellipse{\pgfpoint{10.500000\du}{9.000000\du}}{\pgfpoint{0.500000\du}{0\du}}{\pgfpoint{0\du}{0.500000\du}}
\pgfusepath{stroke}
\pgfsetlinewidth{0.100000\du}
\pgfsetdash{}{0pt}
\definecolor{diafillcolor}{rgb}{1.000000, 1.000000, 1.000000}
\pgfsetfillcolor{diafillcolor}
\pgfsetfillopacity{1.000000}
\pgfpathellipse{\pgfpoint{14.500000\du}{9.000000\du}}{\pgfpoint{0.500000\du}{0\du}}{\pgfpoint{0\du}{0.500000\du}}
\pgfusepath{fill}
\definecolor{dialinecolor}{rgb}{0.000000, 0.000000, 0.000000}
\pgfsetstrokecolor{dialinecolor}
\pgfsetstrokeopacity{1.000000}
\pgfpathellipse{\pgfpoint{14.500000\du}{9.000000\du}}{\pgfpoint{0.500000\du}{0\du}}{\pgfpoint{0\du}{0.500000\du}}
\pgfusepath{stroke}
\pgfsetlinewidth{0.100000\du}
\pgfsetdash{}{0pt}
\pgfsetbuttcap
{
\definecolor{diafillcolor}{rgb}{0.000000, 0.000000, 0.000000}
\pgfsetfillcolor{diafillcolor}
\pgfsetfillopacity{1.000000}
\pgfsetarrowsend{to}
\definecolor{dialinecolor}{rgb}{0.000000, 0.000000, 0.000000}
\pgfsetstrokecolor{dialinecolor}
\pgfsetstrokeopacity{1.000000}
\draw (13.951172\du,9.000000\du)--(11.048828\du,9.000000\du);
}
\pgfsetlinewidth{0.100000\du}
\pgfsetdash{}{0pt}
\definecolor{diafillcolor}{rgb}{1.000000, 1.000000, 1.000000}
\pgfsetfillcolor{diafillcolor}
\pgfsetfillopacity{1.000000}
\pgfpathellipse{\pgfpoint{19.000000\du}{9.000000\du}}{\pgfpoint{0.500000\du}{0\du}}{\pgfpoint{0\du}{0.500000\du}}
\pgfusepath{fill}
\definecolor{dialinecolor}{rgb}{0.000000, 0.000000, 0.000000}
\pgfsetstrokecolor{dialinecolor}
\pgfsetstrokeopacity{1.000000}
\pgfpathellipse{\pgfpoint{19.000000\du}{9.000000\du}}{\pgfpoint{0.500000\du}{0\du}}{\pgfpoint{0\du}{0.500000\du}}
\pgfusepath{stroke}
\pgfsetlinewidth{0.100000\du}
\pgfsetdash{}{0pt}
\pgfsetbuttcap
{
\definecolor{diafillcolor}{rgb}{0.000000, 0.000000, 0.000000}
\pgfsetfillcolor{diafillcolor}
\pgfsetfillopacity{1.000000}
\pgfsetarrowsend{to}
\definecolor{dialinecolor}{rgb}{0.000000, 0.000000, 0.000000}
\pgfsetstrokecolor{dialinecolor}
\pgfsetstrokeopacity{1.000000}
\draw (18.450134\du,9.000000\du)--(15.049866\du,9.000000\du);
}
\definecolor{dialinecolor}{rgb}{0.000000, 0.000000, 0.000000}
\pgfsetstrokecolor{dialinecolor}
\pgfsetstrokeopacity{1.000000}
\definecolor{diafillcolor}{rgb}{0.000000, 0.000000, 0.000000}
\pgfsetfillcolor{diafillcolor}
\pgfsetfillopacity{1.000000}
\node[anchor=base west,inner sep=0pt,outer sep=0pt,color=dialinecolor] at (20.000000\du,11.500000\du){\ldots};
\pgfsetlinewidth{0.100000\du}
\pgfsetdash{}{0pt}
\definecolor{diafillcolor}{rgb}{1.000000, 1.000000, 1.000000}
\pgfsetfillcolor{diafillcolor}
\pgfsetfillopacity{1.000000}
\pgfpathellipse{\pgfpoint{10.500000\du}{11.000000\du}}{\pgfpoint{0.500000\du}{0\du}}{\pgfpoint{0\du}{0.500000\du}}
\pgfusepath{fill}
\definecolor{dialinecolor}{rgb}{0.000000, 0.000000, 0.000000}
\pgfsetstrokecolor{dialinecolor}
\pgfsetstrokeopacity{1.000000}
\pgfpathellipse{\pgfpoint{10.500000\du}{11.000000\du}}{\pgfpoint{0.500000\du}{0\du}}{\pgfpoint{0\du}{0.500000\du}}
\pgfusepath{stroke}
\pgfsetlinewidth{0.100000\du}
\pgfsetdash{}{0pt}
\definecolor{diafillcolor}{rgb}{1.000000, 1.000000, 1.000000}
\pgfsetfillcolor{diafillcolor}
\pgfsetfillopacity{1.000000}
\pgfpathellipse{\pgfpoint{14.500000\du}{11.000000\du}}{\pgfpoint{0.500000\du}{0\du}}{\pgfpoint{0\du}{0.500000\du}}
\pgfusepath{fill}
\definecolor{dialinecolor}{rgb}{0.000000, 0.000000, 0.000000}
\pgfsetstrokecolor{dialinecolor}
\pgfsetstrokeopacity{1.000000}
\pgfpathellipse{\pgfpoint{14.500000\du}{11.000000\du}}{\pgfpoint{0.500000\du}{0\du}}{\pgfpoint{0\du}{0.500000\du}}
\pgfusepath{stroke}
\pgfsetlinewidth{0.100000\du}
\pgfsetdash{}{0pt}
\pgfsetbuttcap
{
\definecolor{diafillcolor}{rgb}{0.000000, 0.000000, 0.000000}
\pgfsetfillcolor{diafillcolor}
\pgfsetfillopacity{1.000000}
\pgfsetarrowsend{to}
\definecolor{dialinecolor}{rgb}{0.000000, 0.000000, 0.000000}
\pgfsetstrokecolor{dialinecolor}
\pgfsetstrokeopacity{1.000000}
\draw (13.951172\du,11.000000\du)--(11.048828\du,11.000000\du);
}
\pgfsetlinewidth{0.100000\du}
\pgfsetdash{}{0pt}
\definecolor{diafillcolor}{rgb}{1.000000, 1.000000, 1.000000}
\pgfsetfillcolor{diafillcolor}
\pgfsetfillopacity{1.000000}
\pgfpathellipse{\pgfpoint{19.000000\du}{11.000000\du}}{\pgfpoint{0.500000\du}{0\du}}{\pgfpoint{0\du}{0.500000\du}}
\pgfusepath{fill}
\definecolor{dialinecolor}{rgb}{0.000000, 0.000000, 0.000000}
\pgfsetstrokecolor{dialinecolor}
\pgfsetstrokeopacity{1.000000}
\pgfpathellipse{\pgfpoint{19.000000\du}{11.000000\du}}{\pgfpoint{0.500000\du}{0\du}}{\pgfpoint{0\du}{0.500000\du}}
\pgfusepath{stroke}
\pgfsetlinewidth{0.100000\du}
\pgfsetdash{}{0pt}
\pgfsetbuttcap
{
\definecolor{diafillcolor}{rgb}{0.000000, 0.000000, 0.000000}
\pgfsetfillcolor{diafillcolor}
\pgfsetfillopacity{1.000000}
\pgfsetarrowsend{to}
\definecolor{dialinecolor}{rgb}{0.000000, 0.000000, 0.000000}
\pgfsetstrokecolor{dialinecolor}
\pgfsetstrokeopacity{1.000000}
\draw (18.450134\du,11.000000\du)--(15.049866\du,11.000000\du);
}
\pgfsetlinewidth{0.100000\du}
\pgfsetdash{}{0pt}
\definecolor{diafillcolor}{rgb}{1.000000, 1.000000, 1.000000}
\pgfsetfillcolor{diafillcolor}
\pgfsetfillopacity{1.000000}
\pgfpathellipse{\pgfpoint{10.500000\du}{14.000000\du}}{\pgfpoint{0.500000\du}{0\du}}{\pgfpoint{0\du}{0.500000\du}}
\pgfusepath{fill}
\definecolor{dialinecolor}{rgb}{0.000000, 0.000000, 0.000000}
\pgfsetstrokecolor{dialinecolor}
\pgfsetstrokeopacity{1.000000}
\pgfpathellipse{\pgfpoint{10.500000\du}{14.000000\du}}{\pgfpoint{0.500000\du}{0\du}}{\pgfpoint{0\du}{0.500000\du}}
\pgfusepath{stroke}
\pgfsetlinewidth{0.100000\du}
\pgfsetdash{}{0pt}
\definecolor{diafillcolor}{rgb}{1.000000, 1.000000, 1.000000}
\pgfsetfillcolor{diafillcolor}
\pgfsetfillopacity{1.000000}
\pgfpathellipse{\pgfpoint{14.500000\du}{14.000000\du}}{\pgfpoint{0.500000\du}{0\du}}{\pgfpoint{0\du}{0.500000\du}}
\pgfusepath{fill}
\definecolor{dialinecolor}{rgb}{0.000000, 0.000000, 0.000000}
\pgfsetstrokecolor{dialinecolor}
\pgfsetstrokeopacity{1.000000}
\pgfpathellipse{\pgfpoint{14.500000\du}{14.000000\du}}{\pgfpoint{0.500000\du}{0\du}}{\pgfpoint{0\du}{0.500000\du}}
\pgfusepath{stroke}
\pgfsetlinewidth{0.100000\du}
\pgfsetdash{}{0pt}
\pgfsetbuttcap
{
\definecolor{diafillcolor}{rgb}{0.000000, 0.000000, 0.000000}
\pgfsetfillcolor{diafillcolor}
\pgfsetfillopacity{1.000000}
\pgfsetarrowsend{to}
\definecolor{dialinecolor}{rgb}{0.000000, 0.000000, 0.000000}
\pgfsetstrokecolor{dialinecolor}
\pgfsetstrokeopacity{1.000000}
\draw (13.951172\du,14.000000\du)--(11.048828\du,14.000000\du);
}
\pgfsetlinewidth{0.100000\du}
\pgfsetdash{}{0pt}
\definecolor{diafillcolor}{rgb}{1.000000, 1.000000, 1.000000}
\pgfsetfillcolor{diafillcolor}
\pgfsetfillopacity{1.000000}
\pgfpathellipse{\pgfpoint{19.000000\du}{14.000000\du}}{\pgfpoint{0.500000\du}{0\du}}{\pgfpoint{0\du}{0.500000\du}}
\pgfusepath{fill}
\definecolor{dialinecolor}{rgb}{0.000000, 0.000000, 0.000000}
\pgfsetstrokecolor{dialinecolor}
\pgfsetstrokeopacity{1.000000}
\pgfpathellipse{\pgfpoint{19.000000\du}{14.000000\du}}{\pgfpoint{0.500000\du}{0\du}}{\pgfpoint{0\du}{0.500000\du}}
\pgfusepath{stroke}
\pgfsetlinewidth{0.100000\du}
\pgfsetdash{}{0pt}
\pgfsetbuttcap
{
\definecolor{diafillcolor}{rgb}{0.000000, 0.000000, 0.000000}
\pgfsetfillcolor{diafillcolor}
\pgfsetfillopacity{1.000000}
\pgfsetarrowsend{to}
\definecolor{dialinecolor}{rgb}{0.000000, 0.000000, 0.000000}
\pgfsetstrokecolor{dialinecolor}
\pgfsetstrokeopacity{1.000000}
\draw (18.450134\du,14.000000\du)--(15.049866\du,14.000000\du);
}
\definecolor{dialinecolor}{rgb}{0.000000, 0.000000, 0.000000}
\pgfsetstrokecolor{dialinecolor}
\pgfsetstrokeopacity{1.000000}
\definecolor{diafillcolor}{rgb}{0.000000, 0.000000, 0.000000}
\pgfsetfillcolor{diafillcolor}
\pgfsetfillopacity{1.000000}
\node[anchor=base west,inner sep=0pt,outer sep=0pt,color=dialinecolor] at (13.500000\du,12.500000\du){\ldots};
\pgfsetlinewidth{0.100000\du}
\pgfsetdash{}{0pt}
\definecolor{diafillcolor}{rgb}{1.000000, 1.000000, 1.000000}
\pgfsetfillcolor{diafillcolor}
\pgfsetfillopacity{1.000000}
\pgfpathellipse{\pgfpoint{10.500000\du}{16.500000\du}}{\pgfpoint{0.500000\du}{0\du}}{\pgfpoint{0\du}{0.500000\du}}
\pgfusepath{fill}
\definecolor{dialinecolor}{rgb}{0.000000, 0.000000, 0.000000}
\pgfsetstrokecolor{dialinecolor}
\pgfsetstrokeopacity{1.000000}
\pgfpathellipse{\pgfpoint{10.500000\du}{16.500000\du}}{\pgfpoint{0.500000\du}{0\du}}{\pgfpoint{0\du}{0.500000\du}}
\pgfusepath{stroke}
\pgfsetlinewidth{0.100000\du}
\pgfsetdash{}{0pt}
\definecolor{diafillcolor}{rgb}{1.000000, 1.000000, 1.000000}
\pgfsetfillcolor{diafillcolor}
\pgfsetfillopacity{1.000000}
\pgfpathellipse{\pgfpoint{14.500000\du}{16.500000\du}}{\pgfpoint{0.500000\du}{0\du}}{\pgfpoint{0\du}{0.500000\du}}
\pgfusepath{fill}
\definecolor{dialinecolor}{rgb}{0.000000, 0.000000, 0.000000}
\pgfsetstrokecolor{dialinecolor}
\pgfsetstrokeopacity{1.000000}
\pgfpathellipse{\pgfpoint{14.500000\du}{16.500000\du}}{\pgfpoint{0.500000\du}{0\du}}{\pgfpoint{0\du}{0.500000\du}}
\pgfusepath{stroke}
\pgfsetlinewidth{0.100000\du}
\pgfsetdash{}{0pt}
\pgfsetbuttcap
{
\definecolor{diafillcolor}{rgb}{0.000000, 0.000000, 0.000000}
\pgfsetfillcolor{diafillcolor}
\pgfsetfillopacity{1.000000}
\pgfsetarrowsend{to}
\definecolor{dialinecolor}{rgb}{0.000000, 0.000000, 0.000000}
\pgfsetstrokecolor{dialinecolor}
\pgfsetstrokeopacity{1.000000}
\draw (13.951172\du,16.500000\du)--(11.048828\du,16.500000\du);
}
\pgfsetlinewidth{0.100000\du}
\pgfsetdash{}{0pt}
\definecolor{diafillcolor}{rgb}{1.000000, 1.000000, 1.000000}
\pgfsetfillcolor{diafillcolor}
\pgfsetfillopacity{1.000000}
\pgfpathellipse{\pgfpoint{19.000000\du}{16.500000\du}}{\pgfpoint{0.500000\du}{0\du}}{\pgfpoint{0\du}{0.500000\du}}
\pgfusepath{fill}
\definecolor{dialinecolor}{rgb}{0.000000, 0.000000, 0.000000}
\pgfsetstrokecolor{dialinecolor}
\pgfsetstrokeopacity{1.000000}
\pgfpathellipse{\pgfpoint{19.000000\du}{16.500000\du}}{\pgfpoint{0.500000\du}{0\du}}{\pgfpoint{0\du}{0.500000\du}}
\pgfusepath{stroke}
\pgfsetlinewidth{0.100000\du}
\pgfsetdash{}{0pt}
\pgfsetbuttcap
{
\definecolor{diafillcolor}{rgb}{0.000000, 0.000000, 0.000000}
\pgfsetfillcolor{diafillcolor}
\pgfsetfillopacity{1.000000}
\pgfsetarrowsend{to}
\definecolor{dialinecolor}{rgb}{0.000000, 0.000000, 0.000000}
\pgfsetstrokecolor{dialinecolor}
\pgfsetstrokeopacity{1.000000}
\draw (18.450134\du,16.500000\du)--(15.049866\du,16.500000\du);
}
\pgfsetlinewidth{0.100000\du}
\pgfsetdash{}{0pt}
\definecolor{diafillcolor}{rgb}{1.000000, 1.000000, 1.000000}
\pgfsetfillcolor{diafillcolor}
\pgfsetfillopacity{1.000000}
\pgfpathellipse{\pgfpoint{10.500000\du}{18.500000\du}}{\pgfpoint{0.500000\du}{0\du}}{\pgfpoint{0\du}{0.500000\du}}
\pgfusepath{fill}
\definecolor{dialinecolor}{rgb}{0.000000, 0.000000, 0.000000}
\pgfsetstrokecolor{dialinecolor}
\pgfsetstrokeopacity{1.000000}
\pgfpathellipse{\pgfpoint{10.500000\du}{18.500000\du}}{\pgfpoint{0.500000\du}{0\du}}{\pgfpoint{0\du}{0.500000\du}}
\pgfusepath{stroke}
\pgfsetlinewidth{0.100000\du}
\pgfsetdash{}{0pt}
\definecolor{diafillcolor}{rgb}{1.000000, 1.000000, 1.000000}
\pgfsetfillcolor{diafillcolor}
\pgfsetfillopacity{1.000000}
\pgfpathellipse{\pgfpoint{14.500000\du}{18.500000\du}}{\pgfpoint{0.500000\du}{0\du}}{\pgfpoint{0\du}{0.500000\du}}
\pgfusepath{fill}
\definecolor{dialinecolor}{rgb}{0.000000, 0.000000, 0.000000}
\pgfsetstrokecolor{dialinecolor}
\pgfsetstrokeopacity{1.000000}
\pgfpathellipse{\pgfpoint{14.500000\du}{18.500000\du}}{\pgfpoint{0.500000\du}{0\du}}{\pgfpoint{0\du}{0.500000\du}}
\pgfusepath{stroke}
\pgfsetlinewidth{0.100000\du}
\pgfsetdash{}{0pt}
\pgfsetbuttcap
{
\definecolor{diafillcolor}{rgb}{0.000000, 0.000000, 0.000000}
\pgfsetfillcolor{diafillcolor}
\pgfsetfillopacity{1.000000}
\pgfsetarrowsend{to}
\definecolor{dialinecolor}{rgb}{0.000000, 0.000000, 0.000000}
\pgfsetstrokecolor{dialinecolor}
\pgfsetstrokeopacity{1.000000}
\draw (13.951172\du,18.500000\du)--(11.048828\du,18.500000\du);
}
\pgfsetlinewidth{0.100000\du}
\pgfsetdash{}{0pt}
\definecolor{diafillcolor}{rgb}{1.000000, 1.000000, 1.000000}
\pgfsetfillcolor{diafillcolor}
\pgfsetfillopacity{1.000000}
\pgfpathellipse{\pgfpoint{19.000000\du}{18.500000\du}}{\pgfpoint{0.500000\du}{0\du}}{\pgfpoint{0\du}{0.500000\du}}
\pgfusepath{fill}
\definecolor{dialinecolor}{rgb}{0.000000, 0.000000, 0.000000}
\pgfsetstrokecolor{dialinecolor}
\pgfsetstrokeopacity{1.000000}
\pgfpathellipse{\pgfpoint{19.000000\du}{18.500000\du}}{\pgfpoint{0.500000\du}{0\du}}{\pgfpoint{0\du}{0.500000\du}}
\pgfusepath{stroke}
\pgfsetlinewidth{0.100000\du}
\pgfsetdash{}{0pt}
\pgfsetbuttcap
{
\definecolor{diafillcolor}{rgb}{0.000000, 0.000000, 0.000000}
\pgfsetfillcolor{diafillcolor}
\pgfsetfillopacity{1.000000}
\pgfsetarrowsend{to}
\definecolor{dialinecolor}{rgb}{0.000000, 0.000000, 0.000000}
\pgfsetstrokecolor{dialinecolor}
\pgfsetstrokeopacity{1.000000}
\draw (18.450134\du,18.500000\du)--(15.049866\du,18.500000\du);
}
\pgfsetlinewidth{0.100000\du}
\pgfsetdash{}{0pt}
\definecolor{diafillcolor}{rgb}{1.000000, 1.000000, 1.000000}
\pgfsetfillcolor{diafillcolor}
\pgfsetfillopacity{1.000000}
\pgfpathellipse{\pgfpoint{10.500000\du}{25.000000\du}}{\pgfpoint{0.500000\du}{0\du}}{\pgfpoint{0\du}{0.500000\du}}
\pgfusepath{fill}
\definecolor{dialinecolor}{rgb}{0.000000, 0.000000, 0.000000}
\pgfsetstrokecolor{dialinecolor}
\pgfsetstrokeopacity{1.000000}
\pgfpathellipse{\pgfpoint{10.500000\du}{25.000000\du}}{\pgfpoint{0.500000\du}{0\du}}{\pgfpoint{0\du}{0.500000\du}}
\pgfusepath{stroke}
\pgfsetlinewidth{0.100000\du}
\pgfsetdash{}{0pt}
\definecolor{diafillcolor}{rgb}{1.000000, 1.000000, 1.000000}
\pgfsetfillcolor{diafillcolor}
\pgfsetfillopacity{1.000000}
\pgfpathellipse{\pgfpoint{14.500000\du}{25.000000\du}}{\pgfpoint{0.500000\du}{0\du}}{\pgfpoint{0\du}{0.500000\du}}
\pgfusepath{fill}
\definecolor{dialinecolor}{rgb}{0.000000, 0.000000, 0.000000}
\pgfsetstrokecolor{dialinecolor}
\pgfsetstrokeopacity{1.000000}
\pgfpathellipse{\pgfpoint{14.500000\du}{25.000000\du}}{\pgfpoint{0.500000\du}{0\du}}{\pgfpoint{0\du}{0.500000\du}}
\pgfusepath{stroke}
\pgfsetlinewidth{0.100000\du}
\pgfsetdash{}{0pt}
\pgfsetbuttcap
{
\definecolor{diafillcolor}{rgb}{0.000000, 0.000000, 0.000000}
\pgfsetfillcolor{diafillcolor}
\pgfsetfillopacity{1.000000}
\pgfsetarrowsend{to}
\definecolor{dialinecolor}{rgb}{0.000000, 0.000000, 0.000000}
\pgfsetstrokecolor{dialinecolor}
\pgfsetstrokeopacity{1.000000}
\draw (13.951172\du,25.000000\du)--(11.048828\du,25.000000\du);
}
\pgfsetlinewidth{0.100000\du}
\pgfsetdash{}{0pt}
\definecolor{diafillcolor}{rgb}{1.000000, 1.000000, 1.000000}
\pgfsetfillcolor{diafillcolor}
\pgfsetfillopacity{1.000000}
\pgfpathellipse{\pgfpoint{19.000000\du}{25.000000\du}}{\pgfpoint{0.500000\du}{0\du}}{\pgfpoint{0\du}{0.500000\du}}
\pgfusepath{fill}
\definecolor{dialinecolor}{rgb}{0.000000, 0.000000, 0.000000}
\pgfsetstrokecolor{dialinecolor}
\pgfsetstrokeopacity{1.000000}
\pgfpathellipse{\pgfpoint{19.000000\du}{25.000000\du}}{\pgfpoint{0.500000\du}{0\du}}{\pgfpoint{0\du}{0.500000\du}}
\pgfusepath{stroke}
\pgfsetlinewidth{0.100000\du}
\pgfsetdash{}{0pt}
\pgfsetbuttcap
{
\definecolor{diafillcolor}{rgb}{0.000000, 0.000000, 0.000000}
\pgfsetfillcolor{diafillcolor}
\pgfsetfillopacity{1.000000}
\pgfsetarrowsend{to}
\definecolor{dialinecolor}{rgb}{0.000000, 0.000000, 0.000000}
\pgfsetstrokecolor{dialinecolor}
\pgfsetstrokeopacity{1.000000}
\draw (18.450134\du,25.000000\du)--(15.049866\du,25.000000\du);
}
\definecolor{dialinecolor}{rgb}{0.000000, 0.000000, 0.000000}
\pgfsetstrokecolor{dialinecolor}
\pgfsetstrokeopacity{1.000000}
\definecolor{diafillcolor}{rgb}{0.000000, 0.000000, 0.000000}
\pgfsetfillcolor{diafillcolor}
\pgfsetfillopacity{1.000000}
\node[anchor=base west,inner sep=0pt,outer sep=0pt,color=dialinecolor] at (13.000000\du,22.000000\du){\ldots};
\pgfsetlinewidth{0.100000\du}
\pgfsetdash{}{0pt}
\pgfsetbuttcap
\pgfsetmiterjoin
\pgfsetlinewidth{0.100000\du}
\pgfsetbuttcap
\pgfsetmiterjoin
\pgfsetdash{}{0pt}
\definecolor{dialinecolor}{rgb}{0.000000, 0.000000, 0.000000}
\pgfsetstrokecolor{dialinecolor}
\pgfsetstrokeopacity{1.000000}
\draw (9.000000\du,9.000000\du)--(8.000000\du,9.000000\du);
\pgfsetbuttcap
\pgfsetmiterjoin
\pgfsetdash{}{0pt}
\definecolor{dialinecolor}{rgb}{0.000000, 0.000000, 0.000000}
\pgfsetstrokecolor{dialinecolor}
\pgfsetstrokeopacity{1.000000}
\draw (9.000000\du,14.000000\du)--(8.000000\du,14.000000\du);
\pgfsetbuttcap
\pgfsetmiterjoin
\pgfsetdash{}{0pt}
\definecolor{dialinecolor}{rgb}{0.000000, 0.000000, 0.000000}
\pgfsetstrokecolor{dialinecolor}
\pgfsetstrokeopacity{1.000000}
\draw (8.000000\du,9.000000\du)--(8.000000\du,14.000000\du);
\pgfsetbuttcap
\pgfsetmiterjoin
\pgfsetdash{}{0pt}
\definecolor{dialinecolor}{rgb}{0.000000, 0.000000, 0.000000}
\pgfsetstrokecolor{dialinecolor}
\pgfsetstrokeopacity{1.000000}
\draw (7.000000\du,11.500000\du)--(8.000000\du,11.500000\du);
\definecolor{dialinecolor}{rgb}{0.000000, 0.000000, 0.000000}
\pgfsetstrokecolor{dialinecolor}
\pgfsetstrokeopacity{1.000000}
\definecolor{diafillcolor}{rgb}{0.000000, 0.000000, 0.000000}
\pgfsetfillcolor{diafillcolor}
\pgfsetfillopacity{1.000000}
\node[anchor=base,inner sep=0pt, outer sep=0pt,color=dialinecolor] at (12.000000\du,11.700000\du){};
\pgfsetlinewidth{0.100000\du}
\pgfsetdash{}{0pt}
\pgfsetbuttcap
\pgfsetmiterjoin
\pgfsetlinewidth{0.100000\du}
\pgfsetbuttcap
\pgfsetmiterjoin
\pgfsetdash{}{0pt}
\definecolor{dialinecolor}{rgb}{0.000000, 0.000000, 0.000000}
\pgfsetstrokecolor{dialinecolor}
\pgfsetstrokeopacity{1.000000}
\draw (9.000000\du,16.500000\du)--(8.000000\du,16.500000\du);
\pgfsetbuttcap
\pgfsetmiterjoin
\pgfsetdash{}{0pt}
\definecolor{dialinecolor}{rgb}{0.000000, 0.000000, 0.000000}
\pgfsetstrokecolor{dialinecolor}
\pgfsetstrokeopacity{1.000000}
\draw (9.000000\du,25.000000\du)--(8.000000\du,25.000000\du);
\pgfsetbuttcap
\pgfsetmiterjoin
\pgfsetdash{}{0pt}
\definecolor{dialinecolor}{rgb}{0.000000, 0.000000, 0.000000}
\pgfsetstrokecolor{dialinecolor}
\pgfsetstrokeopacity{1.000000}
\draw (8.000000\du,16.500000\du)--(8.000000\du,25.000000\du);
\pgfsetbuttcap
\pgfsetmiterjoin
\pgfsetdash{}{0pt}
\definecolor{dialinecolor}{rgb}{0.000000, 0.000000, 0.000000}
\pgfsetstrokecolor{dialinecolor}
\pgfsetstrokeopacity{1.000000}
\draw (7.000000\du,20.750000\du)--(8.000000\du,20.750000\du);
\definecolor{dialinecolor}{rgb}{0.000000, 0.000000, 0.000000}
\pgfsetstrokecolor{dialinecolor}
\pgfsetstrokeopacity{1.000000}
\definecolor{diafillcolor}{rgb}{0.000000, 0.000000, 0.000000}
\pgfsetfillcolor{diafillcolor}
\pgfsetfillopacity{1.000000}
\node[anchor=base,inner sep=0pt, outer sep=0pt,color=dialinecolor] at (12.000000\du,20.950000\du){};
\definecolor{dialinecolor}{rgb}{0.000000, 0.000000, 0.000000}
\pgfsetstrokecolor{dialinecolor}
\pgfsetstrokeopacity{1.000000}
\definecolor{diafillcolor}{rgb}{0.000000, 0.000000, 0.000000}
\pgfsetfillcolor{diafillcolor}
\pgfsetfillopacity{1.000000}
\node[anchor=base west,inner sep=0pt,outer sep=0pt,color=dialinecolor] at (5.000000\du,12.000000\du){$\sqrt{n}$};
\definecolor{dialinecolor}{rgb}{0.000000, 0.000000, 0.000000}
\pgfsetstrokecolor{dialinecolor}
\pgfsetstrokeopacity{1.000000}
\definecolor{diafillcolor}{rgb}{0.000000, 0.000000, 0.000000}
\pgfsetfillcolor{diafillcolor}
\pgfsetfillopacity{1.000000}
\node[anchor=base west,inner sep=0pt,outer sep=0pt,color=dialinecolor] at (4.500000\du,21.000000\du){$\frac{2n}{\log n}$};
\pgfsetlinewidth{0.100000\du}
\pgfsetdash{}{0pt}
\definecolor{diafillcolor}{rgb}{1.000000, 1.000000, 1.000000}
\pgfsetfillcolor{diafillcolor}
\pgfsetfillopacity{1.000000}
\pgfpathellipse{\pgfpoint{24.000000\du}{9.000000\du}}{\pgfpoint{0.500000\du}{0\du}}{\pgfpoint{0\du}{0.500000\du}}
\pgfusepath{fill}
\definecolor{dialinecolor}{rgb}{0.000000, 0.000000, 0.000000}
\pgfsetstrokecolor{dialinecolor}
\pgfsetstrokeopacity{1.000000}
\pgfpathellipse{\pgfpoint{24.000000\du}{9.000000\du}}{\pgfpoint{0.500000\du}{0\du}}{\pgfpoint{0\du}{0.500000\du}}
\pgfusepath{stroke}
\pgfsetlinewidth{0.100000\du}
\pgfsetdash{}{0pt}
\definecolor{diafillcolor}{rgb}{1.000000, 1.000000, 1.000000}
\pgfsetfillcolor{diafillcolor}
\pgfsetfillopacity{1.000000}
\pgfpathellipse{\pgfpoint{28.000000\du}{9.000000\du}}{\pgfpoint{0.500000\du}{0\du}}{\pgfpoint{0\du}{0.500000\du}}
\pgfusepath{fill}
\definecolor{dialinecolor}{rgb}{0.000000, 0.000000, 0.000000}
\pgfsetstrokecolor{dialinecolor}
\pgfsetstrokeopacity{1.000000}
\pgfpathellipse{\pgfpoint{28.000000\du}{9.000000\du}}{\pgfpoint{0.500000\du}{0\du}}{\pgfpoint{0\du}{0.500000\du}}
\pgfusepath{stroke}
\pgfsetlinewidth{0.100000\du}
\pgfsetdash{}{0pt}
\pgfsetbuttcap
{
\definecolor{diafillcolor}{rgb}{0.000000, 0.000000, 0.000000}
\pgfsetfillcolor{diafillcolor}
\pgfsetfillopacity{1.000000}
\pgfsetarrowsend{to}
\definecolor{dialinecolor}{rgb}{0.000000, 0.000000, 0.000000}
\pgfsetstrokecolor{dialinecolor}
\pgfsetstrokeopacity{1.000000}
\draw (27.451172\du,9.000000\du)--(24.548828\du,9.000000\du);
}
\pgfsetlinewidth{0.100000\du}
\pgfsetdash{}{0pt}
\definecolor{diafillcolor}{rgb}{1.000000, 1.000000, 1.000000}
\pgfsetfillcolor{diafillcolor}
\pgfsetfillopacity{1.000000}
\pgfpathellipse{\pgfpoint{32.500000\du}{9.000000\du}}{\pgfpoint{0.500000\du}{0\du}}{\pgfpoint{0\du}{0.500000\du}}
\pgfusepath{fill}
\definecolor{dialinecolor}{rgb}{0.000000, 0.000000, 0.000000}
\pgfsetstrokecolor{dialinecolor}
\pgfsetstrokeopacity{1.000000}
\pgfpathellipse{\pgfpoint{32.500000\du}{9.000000\du}}{\pgfpoint{0.500000\du}{0\du}}{\pgfpoint{0\du}{0.500000\du}}
\pgfusepath{stroke}
\pgfsetlinewidth{0.100000\du}
\pgfsetdash{}{0pt}
\pgfsetbuttcap
{
\definecolor{diafillcolor}{rgb}{0.000000, 0.000000, 0.000000}
\pgfsetfillcolor{diafillcolor}
\pgfsetfillopacity{1.000000}
\pgfsetarrowsend{to}
\definecolor{dialinecolor}{rgb}{0.000000, 0.000000, 0.000000}
\pgfsetstrokecolor{dialinecolor}
\pgfsetstrokeopacity{1.000000}
\draw (31.950134\du,9.000000\du)--(28.549866\du,9.000000\du);
}
\pgfsetlinewidth{0.100000\du}
\pgfsetdash{}{0pt}
\definecolor{diafillcolor}{rgb}{1.000000, 1.000000, 1.000000}
\pgfsetfillcolor{diafillcolor}
\pgfsetfillopacity{1.000000}
\pgfpathellipse{\pgfpoint{24.000000\du}{11.000000\du}}{\pgfpoint{0.500000\du}{0\du}}{\pgfpoint{0\du}{0.500000\du}}
\pgfusepath{fill}
\definecolor{dialinecolor}{rgb}{0.000000, 0.000000, 0.000000}
\pgfsetstrokecolor{dialinecolor}
\pgfsetstrokeopacity{1.000000}
\pgfpathellipse{\pgfpoint{24.000000\du}{11.000000\du}}{\pgfpoint{0.500000\du}{0\du}}{\pgfpoint{0\du}{0.500000\du}}
\pgfusepath{stroke}
\pgfsetlinewidth{0.100000\du}
\pgfsetdash{}{0pt}
\definecolor{diafillcolor}{rgb}{1.000000, 1.000000, 1.000000}
\pgfsetfillcolor{diafillcolor}
\pgfsetfillopacity{1.000000}
\pgfpathellipse{\pgfpoint{28.000000\du}{11.000000\du}}{\pgfpoint{0.500000\du}{0\du}}{\pgfpoint{0\du}{0.500000\du}}
\pgfusepath{fill}
\definecolor{dialinecolor}{rgb}{0.000000, 0.000000, 0.000000}
\pgfsetstrokecolor{dialinecolor}
\pgfsetstrokeopacity{1.000000}
\pgfpathellipse{\pgfpoint{28.000000\du}{11.000000\du}}{\pgfpoint{0.500000\du}{0\du}}{\pgfpoint{0\du}{0.500000\du}}
\pgfusepath{stroke}
\pgfsetlinewidth{0.100000\du}
\pgfsetdash{}{0pt}
\pgfsetbuttcap
{
\definecolor{diafillcolor}{rgb}{0.000000, 0.000000, 0.000000}
\pgfsetfillcolor{diafillcolor}
\pgfsetfillopacity{1.000000}
\pgfsetarrowsend{to}
\definecolor{dialinecolor}{rgb}{0.000000, 0.000000, 0.000000}
\pgfsetstrokecolor{dialinecolor}
\pgfsetstrokeopacity{1.000000}
\draw (27.451172\du,11.000000\du)--(24.548828\du,11.000000\du);
}
\pgfsetlinewidth{0.100000\du}
\pgfsetdash{}{0pt}
\definecolor{diafillcolor}{rgb}{1.000000, 1.000000, 1.000000}
\pgfsetfillcolor{diafillcolor}
\pgfsetfillopacity{1.000000}
\pgfpathellipse{\pgfpoint{32.500000\du}{11.000000\du}}{\pgfpoint{0.500000\du}{0\du}}{\pgfpoint{0\du}{0.500000\du}}
\pgfusepath{fill}
\definecolor{dialinecolor}{rgb}{0.000000, 0.000000, 0.000000}
\pgfsetstrokecolor{dialinecolor}
\pgfsetstrokeopacity{1.000000}
\pgfpathellipse{\pgfpoint{32.500000\du}{11.000000\du}}{\pgfpoint{0.500000\du}{0\du}}{\pgfpoint{0\du}{0.500000\du}}
\pgfusepath{stroke}
\pgfsetlinewidth{0.100000\du}
\pgfsetdash{}{0pt}
\pgfsetbuttcap
{
\definecolor{diafillcolor}{rgb}{0.000000, 0.000000, 0.000000}
\pgfsetfillcolor{diafillcolor}
\pgfsetfillopacity{1.000000}
\pgfsetarrowsend{to}
\definecolor{dialinecolor}{rgb}{0.000000, 0.000000, 0.000000}
\pgfsetstrokecolor{dialinecolor}
\pgfsetstrokeopacity{1.000000}
\draw (31.950134\du,11.000000\du)--(28.549866\du,11.000000\du);
}
\pgfsetlinewidth{0.100000\du}
\pgfsetdash{}{0pt}
\definecolor{diafillcolor}{rgb}{1.000000, 1.000000, 1.000000}
\pgfsetfillcolor{diafillcolor}
\pgfsetfillopacity{1.000000}
\pgfpathellipse{\pgfpoint{24.000000\du}{14.000000\du}}{\pgfpoint{0.500000\du}{0\du}}{\pgfpoint{0\du}{0.500000\du}}
\pgfusepath{fill}
\definecolor{dialinecolor}{rgb}{0.000000, 0.000000, 0.000000}
\pgfsetstrokecolor{dialinecolor}
\pgfsetstrokeopacity{1.000000}
\pgfpathellipse{\pgfpoint{24.000000\du}{14.000000\du}}{\pgfpoint{0.500000\du}{0\du}}{\pgfpoint{0\du}{0.500000\du}}
\pgfusepath{stroke}
\pgfsetlinewidth{0.100000\du}
\pgfsetdash{}{0pt}
\definecolor{diafillcolor}{rgb}{1.000000, 1.000000, 1.000000}
\pgfsetfillcolor{diafillcolor}
\pgfsetfillopacity{1.000000}
\pgfpathellipse{\pgfpoint{28.000000\du}{14.000000\du}}{\pgfpoint{0.500000\du}{0\du}}{\pgfpoint{0\du}{0.500000\du}}
\pgfusepath{fill}
\definecolor{dialinecolor}{rgb}{0.000000, 0.000000, 0.000000}
\pgfsetstrokecolor{dialinecolor}
\pgfsetstrokeopacity{1.000000}
\pgfpathellipse{\pgfpoint{28.000000\du}{14.000000\du}}{\pgfpoint{0.500000\du}{0\du}}{\pgfpoint{0\du}{0.500000\du}}
\pgfusepath{stroke}
\pgfsetlinewidth{0.100000\du}
\pgfsetdash{}{0pt}
\pgfsetbuttcap
{
\definecolor{diafillcolor}{rgb}{0.000000, 0.000000, 0.000000}
\pgfsetfillcolor{diafillcolor}
\pgfsetfillopacity{1.000000}
\pgfsetarrowsend{to}
\definecolor{dialinecolor}{rgb}{0.000000, 0.000000, 0.000000}
\pgfsetstrokecolor{dialinecolor}
\pgfsetstrokeopacity{1.000000}
\draw (27.451172\du,14.000000\du)--(24.548828\du,14.000000\du);
}
\pgfsetlinewidth{0.100000\du}
\pgfsetdash{}{0pt}
\definecolor{diafillcolor}{rgb}{1.000000, 1.000000, 1.000000}
\pgfsetfillcolor{diafillcolor}
\pgfsetfillopacity{1.000000}
\pgfpathellipse{\pgfpoint{32.500000\du}{14.000000\du}}{\pgfpoint{0.500000\du}{0\du}}{\pgfpoint{0\du}{0.500000\du}}
\pgfusepath{fill}
\definecolor{dialinecolor}{rgb}{0.000000, 0.000000, 0.000000}
\pgfsetstrokecolor{dialinecolor}
\pgfsetstrokeopacity{1.000000}
\pgfpathellipse{\pgfpoint{32.500000\du}{14.000000\du}}{\pgfpoint{0.500000\du}{0\du}}{\pgfpoint{0\du}{0.500000\du}}
\pgfusepath{stroke}
\pgfsetlinewidth{0.100000\du}
\pgfsetdash{}{0pt}
\pgfsetbuttcap
{
\definecolor{diafillcolor}{rgb}{0.000000, 0.000000, 0.000000}
\pgfsetfillcolor{diafillcolor}
\pgfsetfillopacity{1.000000}
\pgfsetarrowsend{to}
\definecolor{dialinecolor}{rgb}{0.000000, 0.000000, 0.000000}
\pgfsetstrokecolor{dialinecolor}
\pgfsetstrokeopacity{1.000000}
\draw (31.950134\du,14.000000\du)--(28.549866\du,14.000000\du);
}
\definecolor{dialinecolor}{rgb}{0.000000, 0.000000, 0.000000}
\pgfsetstrokecolor{dialinecolor}
\pgfsetstrokeopacity{1.000000}
\definecolor{diafillcolor}{rgb}{0.000000, 0.000000, 0.000000}
\pgfsetfillcolor{diafillcolor}
\pgfsetfillopacity{1.000000}
\node[anchor=base west,inner sep=0pt,outer sep=0pt,color=dialinecolor] at (27.000000\du,12.500000\du){\ldots};
\pgfsetlinewidth{0.100000\du}
\pgfsetdash{}{0pt}
\definecolor{diafillcolor}{rgb}{1.000000, 1.000000, 1.000000}
\pgfsetfillcolor{diafillcolor}
\pgfsetfillopacity{1.000000}
\pgfpathellipse{\pgfpoint{24.000000\du}{16.500000\du}}{\pgfpoint{0.500000\du}{0\du}}{\pgfpoint{0\du}{0.500000\du}}
\pgfusepath{fill}
\definecolor{dialinecolor}{rgb}{0.000000, 0.000000, 0.000000}
\pgfsetstrokecolor{dialinecolor}
\pgfsetstrokeopacity{1.000000}
\pgfpathellipse{\pgfpoint{24.000000\du}{16.500000\du}}{\pgfpoint{0.500000\du}{0\du}}{\pgfpoint{0\du}{0.500000\du}}
\pgfusepath{stroke}
\pgfsetlinewidth{0.100000\du}
\pgfsetdash{}{0pt}
\definecolor{diafillcolor}{rgb}{1.000000, 1.000000, 1.000000}
\pgfsetfillcolor{diafillcolor}
\pgfsetfillopacity{1.000000}
\pgfpathellipse{\pgfpoint{28.000000\du}{16.500000\du}}{\pgfpoint{0.500000\du}{0\du}}{\pgfpoint{0\du}{0.500000\du}}
\pgfusepath{fill}
\definecolor{dialinecolor}{rgb}{0.000000, 0.000000, 0.000000}
\pgfsetstrokecolor{dialinecolor}
\pgfsetstrokeopacity{1.000000}
\pgfpathellipse{\pgfpoint{28.000000\du}{16.500000\du}}{\pgfpoint{0.500000\du}{0\du}}{\pgfpoint{0\du}{0.500000\du}}
\pgfusepath{stroke}
\pgfsetlinewidth{0.100000\du}
\pgfsetdash{}{0pt}
\pgfsetbuttcap
{
\definecolor{diafillcolor}{rgb}{0.000000, 0.000000, 0.000000}
\pgfsetfillcolor{diafillcolor}
\pgfsetfillopacity{1.000000}
\pgfsetarrowsend{to}
\definecolor{dialinecolor}{rgb}{0.000000, 0.000000, 0.000000}
\pgfsetstrokecolor{dialinecolor}
\pgfsetstrokeopacity{1.000000}
\draw (27.451172\du,16.500000\du)--(24.548828\du,16.500000\du);
}
\pgfsetlinewidth{0.100000\du}
\pgfsetdash{}{0pt}
\definecolor{diafillcolor}{rgb}{1.000000, 1.000000, 1.000000}
\pgfsetfillcolor{diafillcolor}
\pgfsetfillopacity{1.000000}
\pgfpathellipse{\pgfpoint{32.500000\du}{16.500000\du}}{\pgfpoint{0.500000\du}{0\du}}{\pgfpoint{0\du}{0.500000\du}}
\pgfusepath{fill}
\definecolor{dialinecolor}{rgb}{0.000000, 0.000000, 0.000000}
\pgfsetstrokecolor{dialinecolor}
\pgfsetstrokeopacity{1.000000}
\pgfpathellipse{\pgfpoint{32.500000\du}{16.500000\du}}{\pgfpoint{0.500000\du}{0\du}}{\pgfpoint{0\du}{0.500000\du}}
\pgfusepath{stroke}
\pgfsetlinewidth{0.100000\du}
\pgfsetdash{}{0pt}
\pgfsetbuttcap
{
\definecolor{diafillcolor}{rgb}{0.000000, 0.000000, 0.000000}
\pgfsetfillcolor{diafillcolor}
\pgfsetfillopacity{1.000000}
\pgfsetarrowsend{to}
\definecolor{dialinecolor}{rgb}{0.000000, 0.000000, 0.000000}
\pgfsetstrokecolor{dialinecolor}
\pgfsetstrokeopacity{1.000000}
\draw (31.950134\du,16.500000\du)--(28.549866\du,16.500000\du);
}
\pgfsetlinewidth{0.100000\du}
\pgfsetdash{}{0pt}
\definecolor{diafillcolor}{rgb}{1.000000, 1.000000, 1.000000}
\pgfsetfillcolor{diafillcolor}
\pgfsetfillopacity{1.000000}
\pgfpathellipse{\pgfpoint{24.000000\du}{18.500000\du}}{\pgfpoint{0.500000\du}{0\du}}{\pgfpoint{0\du}{0.500000\du}}
\pgfusepath{fill}
\definecolor{dialinecolor}{rgb}{0.000000, 0.000000, 0.000000}
\pgfsetstrokecolor{dialinecolor}
\pgfsetstrokeopacity{1.000000}
\pgfpathellipse{\pgfpoint{24.000000\du}{18.500000\du}}{\pgfpoint{0.500000\du}{0\du}}{\pgfpoint{0\du}{0.500000\du}}
\pgfusepath{stroke}
\pgfsetlinewidth{0.100000\du}
\pgfsetdash{}{0pt}
\definecolor{diafillcolor}{rgb}{1.000000, 1.000000, 1.000000}
\pgfsetfillcolor{diafillcolor}
\pgfsetfillopacity{1.000000}
\pgfpathellipse{\pgfpoint{28.000000\du}{18.500000\du}}{\pgfpoint{0.500000\du}{0\du}}{\pgfpoint{0\du}{0.500000\du}}
\pgfusepath{fill}
\definecolor{dialinecolor}{rgb}{0.000000, 0.000000, 0.000000}
\pgfsetstrokecolor{dialinecolor}
\pgfsetstrokeopacity{1.000000}
\pgfpathellipse{\pgfpoint{28.000000\du}{18.500000\du}}{\pgfpoint{0.500000\du}{0\du}}{\pgfpoint{0\du}{0.500000\du}}
\pgfusepath{stroke}
\pgfsetlinewidth{0.100000\du}
\pgfsetdash{}{0pt}
\pgfsetbuttcap
{
\definecolor{diafillcolor}{rgb}{0.000000, 0.000000, 0.000000}
\pgfsetfillcolor{diafillcolor}
\pgfsetfillopacity{1.000000}
\pgfsetarrowsend{to}
\definecolor{dialinecolor}{rgb}{0.000000, 0.000000, 0.000000}
\pgfsetstrokecolor{dialinecolor}
\pgfsetstrokeopacity{1.000000}
\draw (27.451172\du,18.500000\du)--(24.548828\du,18.500000\du);
}
\pgfsetlinewidth{0.100000\du}
\pgfsetdash{}{0pt}
\definecolor{diafillcolor}{rgb}{1.000000, 1.000000, 1.000000}
\pgfsetfillcolor{diafillcolor}
\pgfsetfillopacity{1.000000}
\pgfpathellipse{\pgfpoint{32.500000\du}{18.500000\du}}{\pgfpoint{0.500000\du}{0\du}}{\pgfpoint{0\du}{0.500000\du}}
\pgfusepath{fill}
\definecolor{dialinecolor}{rgb}{0.000000, 0.000000, 0.000000}
\pgfsetstrokecolor{dialinecolor}
\pgfsetstrokeopacity{1.000000}
\pgfpathellipse{\pgfpoint{32.500000\du}{18.500000\du}}{\pgfpoint{0.500000\du}{0\du}}{\pgfpoint{0\du}{0.500000\du}}
\pgfusepath{stroke}
\pgfsetlinewidth{0.100000\du}
\pgfsetdash{}{0pt}
\pgfsetbuttcap
{
\definecolor{diafillcolor}{rgb}{0.000000, 0.000000, 0.000000}
\pgfsetfillcolor{diafillcolor}
\pgfsetfillopacity{1.000000}
\pgfsetarrowsend{to}
\definecolor{dialinecolor}{rgb}{0.000000, 0.000000, 0.000000}
\pgfsetstrokecolor{dialinecolor}
\pgfsetstrokeopacity{1.000000}
\draw (31.950134\du,18.500000\du)--(28.549866\du,18.500000\du);
}
\pgfsetlinewidth{0.100000\du}
\pgfsetdash{}{0pt}
\definecolor{diafillcolor}{rgb}{1.000000, 1.000000, 1.000000}
\pgfsetfillcolor{diafillcolor}
\pgfsetfillopacity{1.000000}
\pgfpathellipse{\pgfpoint{24.000000\du}{25.000000\du}}{\pgfpoint{0.500000\du}{0\du}}{\pgfpoint{0\du}{0.500000\du}}
\pgfusepath{fill}
\definecolor{dialinecolor}{rgb}{0.000000, 0.000000, 0.000000}
\pgfsetstrokecolor{dialinecolor}
\pgfsetstrokeopacity{1.000000}
\pgfpathellipse{\pgfpoint{24.000000\du}{25.000000\du}}{\pgfpoint{0.500000\du}{0\du}}{\pgfpoint{0\du}{0.500000\du}}
\pgfusepath{stroke}
\pgfsetlinewidth{0.100000\du}
\pgfsetdash{}{0pt}
\definecolor{diafillcolor}{rgb}{1.000000, 1.000000, 1.000000}
\pgfsetfillcolor{diafillcolor}
\pgfsetfillopacity{1.000000}
\pgfpathellipse{\pgfpoint{28.000000\du}{25.000000\du}}{\pgfpoint{0.500000\du}{0\du}}{\pgfpoint{0\du}{0.500000\du}}
\pgfusepath{fill}
\definecolor{dialinecolor}{rgb}{0.000000, 0.000000, 0.000000}
\pgfsetstrokecolor{dialinecolor}
\pgfsetstrokeopacity{1.000000}
\pgfpathellipse{\pgfpoint{28.000000\du}{25.000000\du}}{\pgfpoint{0.500000\du}{0\du}}{\pgfpoint{0\du}{0.500000\du}}
\pgfusepath{stroke}
\pgfsetlinewidth{0.100000\du}
\pgfsetdash{}{0pt}
\pgfsetbuttcap
{
\definecolor{diafillcolor}{rgb}{0.000000, 0.000000, 0.000000}
\pgfsetfillcolor{diafillcolor}
\pgfsetfillopacity{1.000000}
\pgfsetarrowsend{to}
\definecolor{dialinecolor}{rgb}{0.000000, 0.000000, 0.000000}
\pgfsetstrokecolor{dialinecolor}
\pgfsetstrokeopacity{1.000000}
\draw (27.451172\du,25.000000\du)--(24.548828\du,25.000000\du);
}
\pgfsetlinewidth{0.100000\du}
\pgfsetdash{}{0pt}
\definecolor{diafillcolor}{rgb}{1.000000, 1.000000, 1.000000}
\pgfsetfillcolor{diafillcolor}
\pgfsetfillopacity{1.000000}
\pgfpathellipse{\pgfpoint{32.500000\du}{25.000000\du}}{\pgfpoint{0.500000\du}{0\du}}{\pgfpoint{0\du}{0.500000\du}}
\pgfusepath{fill}
\definecolor{dialinecolor}{rgb}{0.000000, 0.000000, 0.000000}
\pgfsetstrokecolor{dialinecolor}
\pgfsetstrokeopacity{1.000000}
\pgfpathellipse{\pgfpoint{32.500000\du}{25.000000\du}}{\pgfpoint{0.500000\du}{0\du}}{\pgfpoint{0\du}{0.500000\du}}
\pgfusepath{stroke}
\pgfsetlinewidth{0.100000\du}
\pgfsetdash{}{0pt}
\pgfsetbuttcap
{
\definecolor{diafillcolor}{rgb}{0.000000, 0.000000, 0.000000}
\pgfsetfillcolor{diafillcolor}
\pgfsetfillopacity{1.000000}
\pgfsetarrowsend{to}
\definecolor{dialinecolor}{rgb}{0.000000, 0.000000, 0.000000}
\pgfsetstrokecolor{dialinecolor}
\pgfsetstrokeopacity{1.000000}
\draw (31.950134\du,25.000000\du)--(28.549866\du,25.000000\du);
}
\definecolor{dialinecolor}{rgb}{0.000000, 0.000000, 0.000000}
\pgfsetstrokecolor{dialinecolor}
\pgfsetstrokeopacity{1.000000}
\definecolor{diafillcolor}{rgb}{0.000000, 0.000000, 0.000000}
\pgfsetfillcolor{diafillcolor}
\pgfsetfillopacity{1.000000}
\node[anchor=base west,inner sep=0pt,outer sep=0pt,color=dialinecolor] at (26.500000\du,22.000000\du){\ldots};
\end{tikzpicture}

%% file: redPw2.tex
\ifx\du\undefined
  \newlength{\du}
\fi
\setlength{\du}{7\unitlength}
\begin{tikzpicture}[even odd rule]
\pgftransformxscale{1.000000}
\pgftransformyscale{-1.000000}
\definecolor{dialinecolor}{rgb}{0.000000, 0.000000, 0.000000}
\pgfsetstrokecolor{dialinecolor}
\pgfsetstrokeopacity{1.000000}
\definecolor{diafillcolor}{rgb}{1.000000, 1.000000, 1.000000}
\pgfsetfillcolor{diafillcolor}
\pgfsetfillopacity{1.000000}
\pgfsetlinewidth{0.100000\du}
\pgfsetdash{}{0pt}
\definecolor{diafillcolor}{rgb}{1.000000, 1.000000, 1.000000}
\pgfsetfillcolor{diafillcolor}
\pgfsetfillopacity{1.000000}
\pgfpathellipse{\pgfpoint{14.000000\du}{10.000000\du}}{\pgfpoint{0.500000\du}{0\du}}{\pgfpoint{0\du}{0.500000\du}}
\pgfusepath{fill}
\definecolor{dialinecolor}{rgb}{0.000000, 0.000000, 0.000000}
\pgfsetstrokecolor{dialinecolor}
\pgfsetstrokeopacity{1.000000}
\pgfpathellipse{\pgfpoint{14.000000\du}{10.000000\du}}{\pgfpoint{0.500000\du}{0\du}}{\pgfpoint{0\du}{0.500000\du}}
\pgfusepath{stroke}
\pgfsetlinewidth{0.100000\du}
\pgfsetdash{}{0pt}
\definecolor{diafillcolor}{rgb}{1.000000, 1.000000, 1.000000}
\pgfsetfillcolor{diafillcolor}
\pgfsetfillopacity{1.000000}
\pgfpathellipse{\pgfpoint{14.000000\du}{13.000000\du}}{\pgfpoint{0.500000\du}{0\du}}{\pgfpoint{0\du}{0.500000\du}}
\pgfusepath{fill}
\definecolor{dialinecolor}{rgb}{0.000000, 0.000000, 0.000000}
\pgfsetstrokecolor{dialinecolor}
\pgfsetstrokeopacity{1.000000}
\pgfpathellipse{\pgfpoint{14.000000\du}{13.000000\du}}{\pgfpoint{0.500000\du}{0\du}}{\pgfpoint{0\du}{0.500000\du}}
\pgfusepath{stroke}
\pgfsetlinewidth{0.100000\du}
\pgfsetdash{}{0pt}
\definecolor{diafillcolor}{rgb}{1.000000, 1.000000, 1.000000}
\pgfsetfillcolor{diafillcolor}
\pgfsetfillopacity{1.000000}
\pgfpathellipse{\pgfpoint{14.000000\du}{16.000000\du}}{\pgfpoint{0.500000\du}{0\du}}{\pgfpoint{0\du}{0.500000\du}}
\pgfusepath{fill}
\definecolor{dialinecolor}{rgb}{0.000000, 0.000000, 0.000000}
\pgfsetstrokecolor{dialinecolor}
\pgfsetstrokeopacity{1.000000}
\pgfpathellipse{\pgfpoint{14.000000\du}{16.000000\du}}{\pgfpoint{0.500000\du}{0\du}}{\pgfpoint{0\du}{0.500000\du}}
\pgfusepath{stroke}
\pgfsetlinewidth{0.100000\du}
\pgfsetdash{}{0pt}
\definecolor{diafillcolor}{rgb}{1.000000, 1.000000, 1.000000}
\pgfsetfillcolor{diafillcolor}
\pgfsetfillopacity{1.000000}
\pgfpathellipse{\pgfpoint{14.000000\du}{19.000000\du}}{\pgfpoint{0.500000\du}{0\du}}{\pgfpoint{0\du}{0.500000\du}}
\pgfusepath{fill}
\definecolor{dialinecolor}{rgb}{0.000000, 0.000000, 0.000000}
\pgfsetstrokecolor{dialinecolor}
\pgfsetstrokeopacity{1.000000}
\pgfpathellipse{\pgfpoint{14.000000\du}{19.000000\du}}{\pgfpoint{0.500000\du}{0\du}}{\pgfpoint{0\du}{0.500000\du}}
\pgfusepath{stroke}
\pgfsetlinewidth{0.100000\du}
\pgfsetdash{}{0pt}
\definecolor{diafillcolor}{rgb}{1.000000, 1.000000, 1.000000}
\pgfsetfillcolor{diafillcolor}
\pgfsetfillopacity{1.000000}
\pgfpathellipse{\pgfpoint{14.000000\du}{22.500000\du}}{\pgfpoint{0.500000\du}{0\du}}{\pgfpoint{0\du}{0.500000\du}}
\pgfusepath{fill}
\definecolor{dialinecolor}{rgb}{0.000000, 0.000000, 0.000000}
\pgfsetstrokecolor{dialinecolor}
\pgfsetstrokeopacity{1.000000}
\pgfpathellipse{\pgfpoint{14.000000\du}{22.500000\du}}{\pgfpoint{0.500000\du}{0\du}}{\pgfpoint{0\du}{0.500000\du}}
\pgfusepath{stroke}
\pgfsetlinewidth{0.100000\du}
\pgfsetdash{}{0pt}
\definecolor{diafillcolor}{rgb}{1.000000, 1.000000, 1.000000}
\pgfsetfillcolor{diafillcolor}
\pgfsetfillopacity{1.000000}
\pgfpathellipse{\pgfpoint{20.000000\du}{10.000000\du}}{\pgfpoint{0.500000\du}{0\du}}{\pgfpoint{0\du}{0.500000\du}}
\pgfusepath{fill}
\definecolor{dialinecolor}{rgb}{0.000000, 0.000000, 0.000000}
\pgfsetstrokecolor{dialinecolor}
\pgfsetstrokeopacity{1.000000}
\pgfpathellipse{\pgfpoint{20.000000\du}{10.000000\du}}{\pgfpoint{0.500000\du}{0\du}}{\pgfpoint{0\du}{0.500000\du}}
\pgfusepath{stroke}
\pgfsetlinewidth{0.100000\du}
\pgfsetdash{}{0pt}
\definecolor{diafillcolor}{rgb}{1.000000, 1.000000, 1.000000}
\pgfsetfillcolor{diafillcolor}
\pgfsetfillopacity{1.000000}
\pgfpathellipse{\pgfpoint{20.000000\du}{13.000000\du}}{\pgfpoint{0.500000\du}{0\du}}{\pgfpoint{0\du}{0.500000\du}}
\pgfusepath{fill}
\definecolor{dialinecolor}{rgb}{0.000000, 0.000000, 0.000000}
\pgfsetstrokecolor{dialinecolor}
\pgfsetstrokeopacity{1.000000}
\pgfpathellipse{\pgfpoint{20.000000\du}{13.000000\du}}{\pgfpoint{0.500000\du}{0\du}}{\pgfpoint{0\du}{0.500000\du}}
\pgfusepath{stroke}
\pgfsetlinewidth{0.100000\du}
\pgfsetdash{}{0pt}
\pgfsetbuttcap
{
\definecolor{diafillcolor}{rgb}{0.000000, 0.000000, 0.000000}
\pgfsetfillcolor{diafillcolor}
\pgfsetfillopacity{1.000000}
\pgfsetarrowsend{to}
\definecolor{dialinecolor}{rgb}{0.000000, 0.000000, 0.000000}
\pgfsetstrokecolor{dialinecolor}
\pgfsetstrokeopacity{1.000000}
\draw (20.488281\du,9.755859\du)--(22.000000\du,9.000000\du);
}
\pgfsetlinewidth{0.100000\du}
\pgfsetdash{}{0pt}
\pgfsetbuttcap
{
\definecolor{diafillcolor}{rgb}{0.000000, 0.000000, 0.000000}
\pgfsetfillcolor{diafillcolor}
\pgfsetfillopacity{1.000000}
\pgfsetarrowsend{to}
\definecolor{dialinecolor}{rgb}{0.000000, 0.000000, 0.000000}
\pgfsetstrokecolor{dialinecolor}
\pgfsetstrokeopacity{1.000000}
\draw (20.533691\du,10.133423\du)--(22.000000\du,10.500000\du);
}
\pgfsetlinewidth{0.100000\du}
\pgfsetdash{}{0pt}
\pgfsetbuttcap
{
\definecolor{diafillcolor}{rgb}{0.000000, 0.000000, 0.000000}
\pgfsetfillcolor{diafillcolor}
\pgfsetfillopacity{1.000000}
\pgfsetarrowsend{to}
\definecolor{dialinecolor}{rgb}{0.000000, 0.000000, 0.000000}
\pgfsetstrokecolor{dialinecolor}
\pgfsetstrokeopacity{1.000000}
\draw (20.488281\du,12.755859\du)--(22.000000\du,12.000000\du);
}
\pgfsetlinewidth{0.100000\du}
\pgfsetdash{}{0pt}
\pgfsetbuttcap
{
\definecolor{diafillcolor}{rgb}{0.000000, 0.000000, 0.000000}
\pgfsetfillcolor{diafillcolor}
\pgfsetfillopacity{1.000000}
\pgfsetarrowsend{to}
\definecolor{dialinecolor}{rgb}{0.000000, 0.000000, 0.000000}
\pgfsetstrokecolor{dialinecolor}
\pgfsetstrokeopacity{1.000000}
\draw (20.488281\du,13.244141\du)--(22.000000\du,14.000000\du);
}
\pgfsetlinewidth{0.100000\du}
\pgfsetdash{}{0pt}
\definecolor{diafillcolor}{rgb}{1.000000, 1.000000, 1.000000}
\pgfsetfillcolor{diafillcolor}
\pgfsetfillopacity{1.000000}
\pgfpathellipse{\pgfpoint{20.000000\du}{16.000000\du}}{\pgfpoint{0.500000\du}{0\du}}{\pgfpoint{0\du}{0.500000\du}}
\pgfusepath{fill}
\definecolor{dialinecolor}{rgb}{0.000000, 0.000000, 0.000000}
\pgfsetstrokecolor{dialinecolor}
\pgfsetstrokeopacity{1.000000}
\pgfpathellipse{\pgfpoint{20.000000\du}{16.000000\du}}{\pgfpoint{0.500000\du}{0\du}}{\pgfpoint{0\du}{0.500000\du}}
\pgfusepath{stroke}
\pgfsetlinewidth{0.100000\du}
\pgfsetdash{}{0pt}
\definecolor{diafillcolor}{rgb}{1.000000, 1.000000, 1.000000}
\pgfsetfillcolor{diafillcolor}
\pgfsetfillopacity{1.000000}
\pgfpathellipse{\pgfpoint{20.000000\du}{19.000000\du}}{\pgfpoint{0.500000\du}{0\du}}{\pgfpoint{0\du}{0.500000\du}}
\pgfusepath{fill}
\definecolor{dialinecolor}{rgb}{0.000000, 0.000000, 0.000000}
\pgfsetstrokecolor{dialinecolor}
\pgfsetstrokeopacity{1.000000}
\pgfpathellipse{\pgfpoint{20.000000\du}{19.000000\du}}{\pgfpoint{0.500000\du}{0\du}}{\pgfpoint{0\du}{0.500000\du}}
\pgfusepath{stroke}
\pgfsetlinewidth{0.100000\du}
\pgfsetdash{}{0pt}
\pgfsetbuttcap
{
\definecolor{diafillcolor}{rgb}{0.000000, 0.000000, 0.000000}
\pgfsetfillcolor{diafillcolor}
\pgfsetfillopacity{1.000000}
\pgfsetarrowsend{to}
\definecolor{dialinecolor}{rgb}{0.000000, 0.000000, 0.000000}
\pgfsetstrokecolor{dialinecolor}
\pgfsetstrokeopacity{1.000000}
\draw (20.488281\du,15.755859\du)--(22.000000\du,15.000000\du);
}
\pgfsetlinewidth{0.100000\du}
\pgfsetdash{}{0pt}
\pgfsetbuttcap
{
\definecolor{diafillcolor}{rgb}{0.000000, 0.000000, 0.000000}
\pgfsetfillcolor{diafillcolor}
\pgfsetfillopacity{1.000000}
\pgfsetarrowsend{to}
\definecolor{dialinecolor}{rgb}{0.000000, 0.000000, 0.000000}
\pgfsetstrokecolor{dialinecolor}
\pgfsetstrokeopacity{1.000000}
\draw (20.533691\du,16.133423\du)--(22.000000\du,16.500000\du);
}
\pgfsetlinewidth{0.100000\du}
\pgfsetdash{}{0pt}
\pgfsetbuttcap
{
\definecolor{diafillcolor}{rgb}{0.000000, 0.000000, 0.000000}
\pgfsetfillcolor{diafillcolor}
\pgfsetfillopacity{1.000000}
\pgfsetarrowsend{to}
\definecolor{dialinecolor}{rgb}{0.000000, 0.000000, 0.000000}
\pgfsetstrokecolor{dialinecolor}
\pgfsetstrokeopacity{1.000000}
\draw (20.488281\du,18.755859\du)--(22.000000\du,18.000000\du);
}
\pgfsetlinewidth{0.100000\du}
\pgfsetdash{}{0pt}
\pgfsetbuttcap
{
\definecolor{diafillcolor}{rgb}{0.000000, 0.000000, 0.000000}
\pgfsetfillcolor{diafillcolor}
\pgfsetfillopacity{1.000000}
\pgfsetarrowsend{to}
\definecolor{dialinecolor}{rgb}{0.000000, 0.000000, 0.000000}
\pgfsetstrokecolor{dialinecolor}
\pgfsetstrokeopacity{1.000000}
\draw (20.488281\du,19.244141\du)--(22.000000\du,20.000000\du);
}
\pgfsetlinewidth{0.100000\du}
\pgfsetdash{}{0pt}
\definecolor{diafillcolor}{rgb}{1.000000, 1.000000, 1.000000}
\pgfsetfillcolor{diafillcolor}
\pgfsetfillopacity{1.000000}
\pgfpathellipse{\pgfpoint{20.000000\du}{22.500000\du}}{\pgfpoint{0.500000\du}{0\du}}{\pgfpoint{0\du}{0.500000\du}}
\pgfusepath{fill}
\definecolor{dialinecolor}{rgb}{0.000000, 0.000000, 0.000000}
\pgfsetstrokecolor{dialinecolor}
\pgfsetstrokeopacity{1.000000}
\pgfpathellipse{\pgfpoint{20.000000\du}{22.500000\du}}{\pgfpoint{0.500000\du}{0\du}}{\pgfpoint{0\du}{0.500000\du}}
\pgfusepath{stroke}
\pgfsetlinewidth{0.100000\du}
\pgfsetdash{}{0pt}
\pgfsetbuttcap
{
\definecolor{diafillcolor}{rgb}{0.000000, 0.000000, 0.000000}
\pgfsetfillcolor{diafillcolor}
\pgfsetfillopacity{1.000000}
\pgfsetarrowsend{to}
\definecolor{dialinecolor}{rgb}{0.000000, 0.000000, 0.000000}
\pgfsetstrokecolor{dialinecolor}
\pgfsetstrokeopacity{1.000000}
\draw (20.488281\du,22.255859\du)--(22.000000\du,21.500000\du);
}
\pgfsetlinewidth{0.100000\du}
\pgfsetdash{}{0pt}
\pgfsetbuttcap
{
\definecolor{diafillcolor}{rgb}{0.000000, 0.000000, 0.000000}
\pgfsetfillcolor{diafillcolor}
\pgfsetfillopacity{1.000000}
\pgfsetarrowsend{to}
\definecolor{dialinecolor}{rgb}{0.000000, 0.000000, 0.000000}
\pgfsetstrokecolor{dialinecolor}
\pgfsetstrokeopacity{1.000000}
\draw (20.488281\du,22.744141\du)--(22.000000\du,23.500000\du);
}
\pgfsetlinewidth{0.200000\du}
\pgfsetdash{}{0pt}
\pgfsetbuttcap
{
\definecolor{diafillcolor}{rgb}{0.000000, 0.000000, 0.000000}
\pgfsetfillcolor{diafillcolor}
\pgfsetfillopacity{1.000000}
\pgfsetarrowsend{to}
\definecolor{dialinecolor}{rgb}{0.000000, 0.000000, 0.000000}
\pgfsetstrokecolor{dialinecolor}
\pgfsetstrokeopacity{1.000000}
\pgfpathmoveto{\pgfpoint{19.450360\du}{9.995400\du}}
\pgfpatharc{315}{226}{3.513704\du and 3.513704\du}
\pgfusepath{stroke}
}
\pgfsetlinewidth{0.100000\du}
\pgfsetdash{{0.300000\du}{0.300000\du}}{0\du}
\pgfsetbuttcap
{
\definecolor{diafillcolor}{rgb}{0.000000, 0.000000, 0.000000}
\pgfsetfillcolor{diafillcolor}
\pgfsetfillopacity{1.000000}
\pgfsetarrowsend{to}
\definecolor{dialinecolor}{rgb}{0.000000, 0.000000, 0.000000}
\pgfsetstrokecolor{dialinecolor}
\pgfsetstrokeopacity{1.000000}
\pgfpathmoveto{\pgfpoint{14.549507\du}{10.004471\du}}
\pgfpatharc{135}{46}{3.513704\du and 3.513704\du}
\pgfusepath{stroke}
}
\definecolor{dialinecolor}{rgb}{0.000000, 0.000000, 0.000000}
\pgfsetstrokecolor{dialinecolor}
\pgfsetstrokeopacity{1.000000}
\definecolor{diafillcolor}{rgb}{0.000000, 0.000000, 0.000000}
\pgfsetfillcolor{diafillcolor}
\pgfsetfillopacity{1.000000}
\node[anchor=base west,inner sep=0pt,outer sep=0pt,color=dialinecolor] at (13.500000\du,8.800000\du){$r_1$};
\node[anchor=base west,inner sep=0pt,outer sep=0pt,color=dialinecolor] at (10.200000\du,8.400000\du){$r_1'$};
\node[anchor=base west,inner sep=0pt,outer sep=0pt,color=dialinecolor] at (10.200000\du,10.800000\du){$r_1''$};
\node[anchor=base west,inner sep=0pt,outer sep=0pt,color=dialinecolor] at (16.500000\du,10.600000\du){-2};
\node[anchor=base west,inner sep=0pt,outer sep=0pt,color=dialinecolor] at (16.500000\du,19.700000\du){-1};
\node[anchor=base west,inner sep=0pt,outer sep=0pt,color=dialinecolor] at (16.500000\du,8.600000\du){2};
\pgfsetlinewidth{0.100000\du}
\pgfsetdash{{0.300000\du}{0.300000\du}}{0\du}
\pgfsetbuttcap
{
\definecolor{diafillcolor}{rgb}{0.000000, 0.000000, 0.000000}
\pgfsetfillcolor{diafillcolor}
\pgfsetfillopacity{1.000000}
\pgfsetarrowsend{to}
\definecolor{dialinecolor}{rgb}{0.000000, 0.000000, 0.000000}
\pgfsetstrokecolor{dialinecolor}
\pgfsetstrokeopacity{1.000000}
\pgfpathmoveto{\pgfpoint{14.499826\du}{12.999834\du}}
\pgfpatharc{134}{47}{3.625000\du and 3.625000\du}
\pgfusepath{stroke}
}
\pgfsetlinewidth{0.100000\du}
\pgfsetdash{{0.300000\du}{0.300000\du}}{0\du}
\pgfsetbuttcap
{
\definecolor{diafillcolor}{rgb}{0.000000, 0.000000, 0.000000}
\pgfsetfillcolor{diafillcolor}
\pgfsetfillopacity{1.000000}
\pgfsetarrowsend{to}
\definecolor{dialinecolor}{rgb}{0.000000, 0.000000, 0.000000}
\pgfsetstrokecolor{dialinecolor}
\pgfsetstrokeopacity{1.000000}
\pgfpathmoveto{\pgfpoint{14.499826\du}{15.999834\du}}
\pgfpatharc{134}{47}{3.625000\du and 3.625000\du}
\pgfusepath{stroke}
}
\pgfsetlinewidth{0.100000\du}
\pgfsetdash{{0.300000\du}{0.300000\du}}{0\du}
\pgfsetbuttcap
{
\definecolor{diafillcolor}{rgb}{0.000000, 0.000000, 0.000000}
\pgfsetfillcolor{diafillcolor}
\pgfsetfillopacity{1.000000}
\pgfsetarrowsend{to}
\definecolor{dialinecolor}{rgb}{0.000000, 0.000000, 0.000000}
\pgfsetstrokecolor{dialinecolor}
\pgfsetstrokeopacity{1.000000}
\pgfpathmoveto{\pgfpoint{14.499826\du}{18.999834\du}}
\pgfpatharc{134}{47}{3.625000\du and 3.625000\du}
\pgfusepath{stroke}
}
\pgfsetlinewidth{0.100000\du}
\pgfsetdash{{0.300000\du}{0.300000\du}}{0\du}
\pgfsetbuttcap
{
\definecolor{diafillcolor}{rgb}{0.000000, 0.000000, 0.000000}
\pgfsetfillcolor{diafillcolor}
\pgfsetfillopacity{1.000000}
\pgfsetarrowsend{to}
\definecolor{dialinecolor}{rgb}{0.000000, 0.000000, 0.000000}
\pgfsetstrokecolor{dialinecolor}
\pgfsetstrokeopacity{1.000000}
\pgfpathmoveto{\pgfpoint{14.499834\du}{22.499838\du}}
\pgfpatharc{135}{46}{3.513704\du and 3.513704\du}
\pgfusepath{stroke}
}
\pgfsetlinewidth{0.200000\du}
\pgfsetdash{}{0pt}
\pgfsetbuttcap
{
\definecolor{diafillcolor}{rgb}{0.000000, 0.000000, 0.000000}
\pgfsetfillcolor{diafillcolor}
\pgfsetfillopacity{1.000000}
\pgfsetarrowsend{to}
\definecolor{dialinecolor}{rgb}{0.000000, 0.000000, 0.000000}
\pgfsetstrokecolor{dialinecolor}
\pgfsetstrokeopacity{1.000000}
\pgfpathmoveto{\pgfpoint{19.500036\du}{13.000034\du}}
\pgfpatharc{314}{227}{3.625000\du and 3.625000\du}
\pgfusepath{stroke}
}
\pgfsetlinewidth{0.200000\du}
\pgfsetdash{}{0pt}
\pgfsetbuttcap
{
\definecolor{diafillcolor}{rgb}{0.000000, 0.000000, 0.000000}
\pgfsetfillcolor{diafillcolor}
\pgfsetfillopacity{1.000000}
\pgfsetarrowsend{to}
\definecolor{dialinecolor}{rgb}{0.000000, 0.000000, 0.000000}
\pgfsetstrokecolor{dialinecolor}
\pgfsetstrokeopacity{1.000000}
\pgfpathmoveto{\pgfpoint{19.500036\du}{16.000034\du}}
\pgfpatharc{314}{227}{3.625000\du and 3.625000\du}
\pgfusepath{stroke}
}
\pgfsetlinewidth{0.200000\du}
\pgfsetdash{}{0pt}
\pgfsetbuttcap
{
\definecolor{diafillcolor}{rgb}{0.000000, 0.000000, 0.000000}
\pgfsetfillcolor{diafillcolor}
\pgfsetfillopacity{1.000000}
\pgfsetarrowsend{to}
\definecolor{dialinecolor}{rgb}{0.000000, 0.000000, 0.000000}
\pgfsetstrokecolor{dialinecolor}
\pgfsetstrokeopacity{1.000000}
\pgfpathmoveto{\pgfpoint{19.500036\du}{19.000034\du}}
\pgfpatharc{314}{227}{3.625000\du and 3.625000\du}
\pgfusepath{stroke}
}
\pgfsetlinewidth{0.200000\du}
\pgfsetdash{}{0pt}
\pgfsetbuttcap
{
\definecolor{diafillcolor}{rgb}{0.000000, 0.000000, 0.000000}
\pgfsetfillcolor{diafillcolor}
\pgfsetfillopacity{1.000000}
\pgfsetarrowsend{to}
\definecolor{dialinecolor}{rgb}{0.000000, 0.000000, 0.000000}
\pgfsetstrokecolor{dialinecolor}
\pgfsetstrokeopacity{1.000000}
\pgfpathmoveto{\pgfpoint{19.500036\du}{22.500034\du}}
\pgfpatharc{314}{227}{3.625000\du and 3.625000\du}
\pgfusepath{stroke}
}
\pgfsetlinewidth{0.200000\du}
\pgfsetdash{}{0pt}
\pgfsetbuttcap
{
\definecolor{diafillcolor}{rgb}{0.000000, 0.000000, 0.000000}
\pgfsetfillcolor{diafillcolor}
\pgfsetfillopacity{1.000000}
\pgfsetarrowsend{to}
\definecolor{dialinecolor}{rgb}{0.000000, 0.000000, 0.000000}
\pgfsetstrokecolor{dialinecolor}
\pgfsetstrokeopacity{1.000000}
\pgfpathmoveto{\pgfpoint{13.991896\du}{19.550108\du}}
\pgfpatharc{260}{101}{1.221700\du and 1.221700\du}
\pgfusepath{stroke}
}
\pgfsetlinewidth{0.200000\du}
\pgfsetdash{}{0pt}
\pgfsetbuttcap
{
\definecolor{diafillcolor}{rgb}{0.000000, 0.000000, 0.000000}
\pgfsetfillcolor{diafillcolor}
\pgfsetfillopacity{1.000000}
\pgfsetarrowsend{to}
\definecolor{dialinecolor}{rgb}{0.000000, 0.000000, 0.000000}
\pgfsetstrokecolor{dialinecolor}
\pgfsetstrokeopacity{1.000000}
\pgfpathmoveto{\pgfpoint{13.987866\du}{16.549696\du}}
\pgfpatharc{271}{89}{0.975348\du and 0.975348\du}
\pgfusepath{stroke}
}
\pgfsetlinewidth{0.200000\du}
\pgfsetdash{}{0pt}
\pgfsetbuttcap
{
\definecolor{diafillcolor}{rgb}{0.000000, 0.000000, 0.000000}
\pgfsetfillcolor{diafillcolor}
\pgfsetfillopacity{1.000000}
\pgfsetarrowsend{to}
\definecolor{dialinecolor}{rgb}{0.000000, 0.000000, 0.000000}
\pgfsetstrokecolor{dialinecolor}
\pgfsetstrokeopacity{1.000000}
\pgfpathmoveto{\pgfpoint{13.987866\du}{13.549696\du}}
\pgfpatharc{271}{89}{0.975348\du and 0.975348\du}
\pgfusepath{stroke}
}
\pgfsetlinewidth{0.200000\du}
\pgfsetdash{}{0pt}
\pgfsetbuttcap
{
\definecolor{diafillcolor}{rgb}{0.000000, 0.000000, 0.000000}
\pgfsetfillcolor{diafillcolor}
\pgfsetfillopacity{1.000000}
\pgfsetarrowsend{to}
\definecolor{dialinecolor}{rgb}{0.000000, 0.000000, 0.000000}
\pgfsetstrokecolor{dialinecolor}
\pgfsetstrokeopacity{1.000000}
\pgfpathmoveto{\pgfpoint{13.987866\du}{10.549696\du}}
\pgfpatharc{271}{89}{0.975348\du and 0.975348\du}
\pgfusepath{stroke}
}
\pgfsetlinewidth{0.100000\du}
\pgfsetdash{{0.300000\du}{0.300000\du}}{0\du}
\pgfsetbuttcap
{
\definecolor{diafillcolor}{rgb}{0.000000, 0.000000, 0.000000}
\pgfsetfillcolor{diafillcolor}
\pgfsetfillopacity{1.000000}
\pgfsetarrowsend{to}
\definecolor{dialinecolor}{rgb}{0.000000, 0.000000, 0.000000}
\pgfsetstrokecolor{dialinecolor}
\pgfsetstrokeopacity{1.000000}
\pgfpathmoveto{\pgfpoint{14.007641\du}{12.449735\du}}
\pgfpatharc{414}{306}{1.162174\du and 1.162174\du}
\pgfusepath{stroke}
}
\pgfsetlinewidth{0.100000\du}
\pgfsetdash{{0.300000\du}{0.300000\du}}{0\du}
\pgfsetbuttcap
{
\definecolor{diafillcolor}{rgb}{0.000000, 0.000000, 0.000000}
\pgfsetfillcolor{diafillcolor}
\pgfsetfillopacity{1.000000}
\pgfsetarrowsend{to}
\definecolor{dialinecolor}{rgb}{0.000000, 0.000000, 0.000000}
\pgfsetstrokecolor{dialinecolor}
\pgfsetstrokeopacity{1.000000}
\pgfpathmoveto{\pgfpoint{14.000016\du}{15.499988\du}}
\pgfpatharc{413}{307}{1.250000\du and 1.250000\du}
\pgfusepath{stroke}
}
\pgfsetlinewidth{0.100000\du}
\pgfsetdash{{0.300000\du}{0.300000\du}}{0\du}
\pgfsetbuttcap
{
\definecolor{diafillcolor}{rgb}{0.000000, 0.000000, 0.000000}
\pgfsetfillcolor{diafillcolor}
\pgfsetfillopacity{1.000000}
\pgfsetarrowsend{to}
\definecolor{dialinecolor}{rgb}{0.000000, 0.000000, 0.000000}
\pgfsetstrokecolor{dialinecolor}
\pgfsetstrokeopacity{1.000000}
\pgfpathmoveto{\pgfpoint{14.000016\du}{18.499988\du}}
\pgfpatharc{413}{307}{1.250000\du and 1.250000\du}
\pgfusepath{stroke}
}
\pgfsetlinewidth{0.100000\du}
\pgfsetdash{{0.300000\du}{0.300000\du}}{0\du}
\pgfsetbuttcap
{
\definecolor{diafillcolor}{rgb}{0.000000, 0.000000, 0.000000}
\pgfsetfillcolor{diafillcolor}
\pgfsetfillopacity{1.000000}
\pgfsetarrowsend{to}
\definecolor{dialinecolor}{rgb}{0.000000, 0.000000, 0.000000}
\pgfsetstrokecolor{dialinecolor}
\pgfsetstrokeopacity{1.000000}
\pgfpathmoveto{\pgfpoint{14.000016\du}{21.999983\du}}
\pgfpatharc{403}{317}{1.812500\du and 1.812500\du}
\pgfusepath{stroke}
}
\definecolor{dialinecolor}{rgb}{0.000000, 0.000000, 0.000000}
\pgfsetstrokecolor{dialinecolor}
\pgfsetstrokeopacity{1.000000}
\definecolor{diafillcolor}{rgb}{0.000000, 0.000000, 0.000000}
\pgfsetfillcolor{diafillcolor}
\pgfsetfillopacity{1.000000}
\node[anchor=base west,inner sep=0pt,outer sep=0pt,color=dialinecolor] at (16.500000\du,24.000000\du){-1};
\pgfsetlinewidth{0.100000\du}
\pgfsetdash{}{0pt}
\definecolor{diafillcolor}{rgb}{1.000000, 1.000000, 1.000000}
\pgfsetfillcolor{diafillcolor}
\pgfsetfillopacity{1.000000}
\pgfpathellipse{\pgfpoint{10.500000\du}{21.500000\du}}{\pgfpoint{0.500000\du}{0\du}}{\pgfpoint{0\du}{0.500000\du}}
\pgfusepath{fill}
\definecolor{dialinecolor}{rgb}{0.000000, 0.000000, 0.000000}
\pgfsetstrokecolor{dialinecolor}
\pgfsetstrokeopacity{1.000000}
\pgfpathellipse{\pgfpoint{10.500000\du}{21.500000\du}}{\pgfpoint{0.500000\du}{0\du}}{\pgfpoint{0\du}{0.500000\du}}
\pgfusepath{stroke}
\pgfsetlinewidth{0.100000\du}
\pgfsetdash{}{0pt}
\definecolor{diafillcolor}{rgb}{1.000000, 1.000000, 1.000000}
\pgfsetfillcolor{diafillcolor}
\pgfsetfillopacity{1.000000}
\pgfpathellipse{\pgfpoint{10.500000\du}{24.000000\du}}{\pgfpoint{0.500000\du}{0\du}}{\pgfpoint{0\du}{0.500000\du}}
\pgfusepath{fill}
\definecolor{dialinecolor}{rgb}{0.000000, 0.000000, 0.000000}
\pgfsetstrokecolor{dialinecolor}
\pgfsetstrokeopacity{1.000000}
\pgfpathellipse{\pgfpoint{10.500000\du}{24.000000\du}}{\pgfpoint{0.500000\du}{0\du}}{\pgfpoint{0\du}{0.500000\du}}
\pgfusepath{stroke}
\pgfsetlinewidth{0.100000\du}
\pgfsetdash{}{0pt}
\pgfsetbuttcap
{
\definecolor{diafillcolor}{rgb}{0.000000, 0.000000, 0.000000}
\pgfsetfillcolor{diafillcolor}
\pgfsetfillopacity{1.000000}
\pgfsetarrowsend{to}
\definecolor{dialinecolor}{rgb}{0.000000, 0.000000, 0.000000}
\pgfsetstrokecolor{dialinecolor}
\pgfsetstrokeopacity{1.000000}
\draw (13.471924\du,22.349121\du)--(11.028076\du,21.650879\du);
}
\pgfsetlinewidth{0.200000\du}
\pgfsetdash{{0.300000\du}{0.300000\du}}{0\du}
\pgfsetbuttcap
{
\definecolor{diafillcolor}{rgb}{0.000000, 0.000000, 0.000000}
\pgfsetfillcolor{diafillcolor}
\pgfsetfillopacity{1.000000}
\pgfsetarrowsend{to}
\definecolor{dialinecolor}{rgb}{0.000000, 0.000000, 0.000000}
\pgfsetstrokecolor{dialinecolor}
\pgfsetstrokeopacity{1.000000}
\draw (13.500977\du,22.713867\du)--(10.999023\du,23.786133\du);
}
\definecolor{dialinecolor}{rgb}{0.000000, 0.000000, 0.000000}
\pgfsetstrokecolor{dialinecolor}
\pgfsetstrokeopacity{1.000000}
\definecolor{diafillcolor}{rgb}{0.000000, 0.000000, 0.000000}
\pgfsetfillcolor{diafillcolor}
\pgfsetfillopacity{1.000000}
\node[anchor=base west,inner sep=0pt,outer sep=0pt,color=dialinecolor] at (12.500000\du,24.000000\du){-2};
\pgfsetlinewidth{0.100000\du}
\pgfsetdash{}{0pt}
\definecolor{diafillcolor}{rgb}{1.000000, 1.000000, 1.000000}
\pgfsetfillcolor{diafillcolor}
\pgfsetfillopacity{1.000000}
\pgfpathellipse{\pgfpoint{10.500000\du}{9.000000\du}}{\pgfpoint{0.500000\du}{0\du}}{\pgfpoint{0\du}{0.500000\du}}
\pgfusepath{fill}
\definecolor{dialinecolor}{rgb}{0.000000, 0.000000, 0.000000}
\pgfsetstrokecolor{dialinecolor}
\pgfsetstrokeopacity{1.000000}
\pgfpathellipse{\pgfpoint{10.500000\du}{9.000000\du}}{\pgfpoint{0.500000\du}{0\du}}{\pgfpoint{0\du}{0.500000\du}}
\pgfusepath{stroke}
\pgfsetlinewidth{0.100000\du}
\pgfsetdash{}{0pt}
\definecolor{diafillcolor}{rgb}{1.000000, 1.000000, 1.000000}
\pgfsetfillcolor{diafillcolor}
\pgfsetfillopacity{1.000000}
\pgfpathellipse{\pgfpoint{10.500000\du}{11.500000\du}}{\pgfpoint{0.500000\du}{0\du}}{\pgfpoint{0\du}{0.500000\du}}
\pgfusepath{fill}
\definecolor{dialinecolor}{rgb}{0.000000, 0.000000, 0.000000}
\pgfsetstrokecolor{dialinecolor}
\pgfsetstrokeopacity{1.000000}
\pgfpathellipse{\pgfpoint{10.500000\du}{11.500000\du}}{\pgfpoint{0.500000\du}{0\du}}{\pgfpoint{0\du}{0.500000\du}}
\pgfusepath{stroke}
\pgfsetlinewidth{0.100000\du}
\pgfsetdash{}{0pt}
\pgfsetbuttcap
{
\definecolor{diafillcolor}{rgb}{0.000000, 0.000000, 0.000000}
\pgfsetfillcolor{diafillcolor}
\pgfsetfillopacity{1.000000}
\pgfsetarrowsend{to}
\definecolor{dialinecolor}{rgb}{0.000000, 0.000000, 0.000000}
\pgfsetstrokecolor{dialinecolor}
\pgfsetstrokeopacity{1.000000}
\draw (13.500000\du,10.000000\du)--(11.000000\du,9.000000\du);
}
\pgfsetlinewidth{0.200000\du}
\pgfsetdash{{0.300000\du}{0.300000\du}}{0\du}
\pgfsetbuttcap
{
\definecolor{diafillcolor}{rgb}{0.000000, 0.000000, 0.000000}
\pgfsetfillcolor{diafillcolor}
\pgfsetfillopacity{1.000000}
\pgfsetarrowsend{to}
\definecolor{dialinecolor}{rgb}{0.000000, 0.000000, 0.000000}
\pgfsetstrokecolor{dialinecolor}
\pgfsetstrokeopacity{1.000000}
\draw (13.500000\du,10.000000\du)--(11.000000\du,11.500000\du);
}
\pgfsetlinewidth{0.100000\du}
\pgfsetdash{}{0pt}
\definecolor{diafillcolor}{rgb}{1.000000, 1.000000, 1.000000}
\pgfsetfillcolor{diafillcolor}
\pgfsetfillopacity{1.000000}
\pgfpathellipse{\pgfpoint{10.500000\du}{15.000000\du}}{\pgfpoint{0.500000\du}{0\du}}{\pgfpoint{0\du}{0.500000\du}}
\pgfusepath{fill}
\definecolor{dialinecolor}{rgb}{0.000000, 0.000000, 0.000000}
\pgfsetstrokecolor{dialinecolor}
\pgfsetstrokeopacity{1.000000}
\pgfpathellipse{\pgfpoint{10.500000\du}{15.000000\du}}{\pgfpoint{0.500000\du}{0\du}}{\pgfpoint{0\du}{0.500000\du}}
\pgfusepath{stroke}
\pgfsetlinewidth{0.100000\du}
\pgfsetdash{}{0pt}
\definecolor{diafillcolor}{rgb}{1.000000, 1.000000, 1.000000}
\pgfsetfillcolor{diafillcolor}
\pgfsetfillopacity{1.000000}
\pgfpathellipse{\pgfpoint{10.500000\du}{17.500000\du}}{\pgfpoint{0.500000\du}{0\du}}{\pgfpoint{0\du}{0.500000\du}}
\pgfusepath{fill}
\definecolor{dialinecolor}{rgb}{0.000000, 0.000000, 0.000000}
\pgfsetstrokecolor{dialinecolor}
\pgfsetstrokeopacity{1.000000}
\pgfpathellipse{\pgfpoint{10.500000\du}{17.500000\du}}{\pgfpoint{0.500000\du}{0\du}}{\pgfpoint{0\du}{0.500000\du}}
\pgfusepath{stroke}
\pgfsetlinewidth{0.100000\du}
\pgfsetdash{}{0pt}
\pgfsetbuttcap
{
\definecolor{diafillcolor}{rgb}{0.000000, 0.000000, 0.000000}
\pgfsetfillcolor{diafillcolor}
\pgfsetfillopacity{1.000000}
\pgfsetarrowsend{to}
\definecolor{dialinecolor}{rgb}{0.000000, 0.000000, 0.000000}
\pgfsetstrokecolor{dialinecolor}
\pgfsetstrokeopacity{1.000000}
\draw (13.500000\du,16.000000\du)--(11.000000\du,15.000000\du);
}
\pgfsetlinewidth{0.200000\du}
\pgfsetdash{{0.300000\du}{0.300000\du}}{0\du}
\pgfsetbuttcap
{
\definecolor{diafillcolor}{rgb}{0.000000, 0.000000, 0.000000}
\pgfsetfillcolor{diafillcolor}
\pgfsetfillopacity{1.000000}
\pgfsetarrowsend{to}
\definecolor{dialinecolor}{rgb}{0.000000, 0.000000, 0.000000}
\pgfsetstrokecolor{dialinecolor}
\pgfsetstrokeopacity{1.000000}
\draw (13.500000\du,16.000000\du)--(11.000000\du,17.500000\du);
}
\end{tikzpicture}

%% file: hedonic.bbl
\begin{thebibliography}{10}

\bibitem{AloisioFV20}
Alessandro Aloisio, Michele Flammini, and Cosimo Vinci.
\newblock The impact of selfishness in hypergraph hedonic games.
\newblock In {\em The Thirty-Fourth {AAAI} Conference on Artificial
  Intelligence, {AAAI} 2020, The Thirty-Second Innovative Applications of
  Artificial Intelligence Conference, {IAAI} 2020, The Tenth {AAAI} Symposium
  on Educational Advances in Artificial Intelligence, {EAAI} 2020, New York,
  NY, USA, February 7-12, 2020}, pages 1766--1773. {AAAI} Press, 2020.
\newblock URL: \url{https://aaai.org/ojs/index.php/AAAI/article/view/5542}.

\bibitem{AzizBBHOP19}
Haris Aziz, Florian Brandl, Felix Brandt, Paul Harrenstein, Martin Olsen, and
  Dominik Peters.
\newblock Fractional hedonic games.
\newblock {\em {ACM} Trans. Economics and Comput.}, 7(2):6:1--6:29, 2019.
\newblock \href {https://doi.org/10.1145/3327970} {\path{doi:10.1145/3327970}}.

\bibitem{AzizBS13}
Haris Aziz, Felix Brandt, and Hans~Georg Seedig.
\newblock Computing desirable partitions in additively separable hedonic games.
\newblock {\em Artif. Intell.}, 195:316--334, 2013.

\bibitem{AzizS16}
Haris Aziz and Rahul Savani.
\newblock Hedonic games.
\newblock In {\em Handbook of Computational Social Choice}, pages 356--376.
  Cambridge University Press, 2016.

\bibitem{Ballester04}
Coralio Ballester.
\newblock {NP}-completeness in hedonic games.
\newblock {\em Games Econ. Behav.}, 49(1):1--30, 2004.

\bibitem{BarrotOSY19}
Nathana{\"{e}}l Barrot, Kazunori Ota, Yuko Sakurai, and Makoto Yokoo.
\newblock Unknown agents in friends oriented hedonic games: Stability and
  complexity.
\newblock In {\em The Thirty-Third {AAAI} Conference on Artificial
  Intelligence, {AAAI} 2019, The Thirty-First Innovative Applications of
  Artificial Intelligence Conference, {IAAI} 2019, The Ninth {AAAI} Symposium
  on Educational Advances in Artificial Intelligence, {EAAI} 2019, Honolulu,
  Hawaii, USA, January 27 - February 1, 2019}, pages 1756--1763. {AAAI} Press,
  2019.
\newblock \href {https://doi.org/10.1609/aaai.v33i01.33011756}
  {\path{doi:10.1609/aaai.v33i01.33011756}}.

\bibitem{BarrotY19}
Nathana{\"{e}}l Barrot and Makoto Yokoo.
\newblock Stable and envy-free partitions in hedonic games.
\newblock In Sarit Kraus, editor, {\em Proceedings of the Twenty-Eighth
  International Joint Conference on Artificial Intelligence, {IJCAI} 2019,
  Macao, China, August 10-16, 2019}, pages 67--73. ijcai.org, 2019.
\newblock \href {https://doi.org/10.24963/ijcai.2019/10}
  {\path{doi:10.24963/ijcai.2019/10}}.

\bibitem{BiloGM19}
Vittorio Bil{\`{o}}, Laurent Gourv{\`{e}}s, and J{\'{e}}r{\^{o}}me Monnot.
\newblock On a simple hedonic game with graph-restricted communication.
\newblock In {\em {SAGT}}, volume 11801 of {\em Lecture Notes in Computer
  Science}, pages 252--265. Springer, 2019.

\bibitem{BoehmerE20}
Niclas Boehmer and Edith Elkind.
\newblock Individual-based stability in hedonic diversity games.
\newblock In {\em The Thirty-Fourth {AAAI} Conference on Artificial
  Intelligence, {AAAI} 2020, The Thirty-Second Innovative Applications of
  Artificial Intelligence Conference, {IAAI} 2020, The Tenth {AAAI} Symposium
  on Educational Advances in Artificial Intelligence, {EAAI} 2020, New York,
  NY, USA, February 7-12, 2020}, pages 1822--1829. {AAAI} Press, 2020.
\newblock URL: \url{https://aaai.org/ojs/index.php/AAAI/article/view/5549}.

\bibitem{0001BW21}
Felix Brandt, Martin Bullinger, and Ana{\"{e}}lle Wilczynski.
\newblock Reaching individually stable coalition structures in hedonic games.
\newblock In {\em Thirty-Fifth {AAAI} Conference on Artificial Intelligence,
  {AAAI} 2021, Thirty-Third Conference on Innovative Applications of Artificial
  Intelligence, {IAAI} 2021, The Eleventh Symposium on Educational Advances in
  Artificial Intelligence, {EAAI} 2021, Virtual Event, February 2-9, 2021},
  pages 5211--5218. {AAAI} Press, 2021.
\newblock URL: \url{https://ojs.aaai.org/index.php/AAAI/article/view/16658}.

\bibitem{BranzeiL09}
Simina Br{\^{a}}nzei and Kate Larson.
\newblock Coalitional affinity games and the stability gap.
\newblock In {\em {IJCAI}}, pages 79--84, 2009.

\bibitem{BullingerK21}
Martin Bullinger and Stefan Kober.
\newblock Loyalty in cardinal hedonic games.
\newblock In Zhi{-}Hua Zhou, editor, {\em Proceedings of the Thirtieth
  International Joint Conference on Artificial Intelligence, {IJCAI} 2021,
  Virtual Event / Montreal, Canada, 19-27 August 2021}, pages 66--72.
  ijcai.org, 2021.
\newblock \href {https://doi.org/10.24963/ijcai.2021/10}
  {\path{doi:10.24963/ijcai.2021/10}}.

\bibitem{Cechlarova16}
Katar{\'{\i}}na Cechl{\'{a}}rov{\'{a}}.
\newblock Stable partition problem.
\newblock In {\em Encyclopedia of Algorithms}, pages 2075--2078. Springer,
  2016.

\bibitem{CyganFKLMPPS15}
Marek Cygan, Fedor~V. Fomin, Lukasz Kowalik, Daniel Lokshtanov, D{\'{a}}niel
  Marx, Marcin Pilipczuk, Michal Pilipczuk, and Saket Saurabh.
\newblock {\em Parameterized Algorithms}.
\newblock Springer, 2015.

\bibitem{DarmannEKLSW18}
Andreas Darmann, Edith Elkind, Sascha Kurz, J{\'{e}}r{\^{o}}me Lang, Joachim
  Schauer, and Gerhard~J. Woeginger.
\newblock Group activity selection problem with approval preferences.
\newblock {\em Int. J. Game Theory}, 47(3):767--796, 2018.
\newblock \href {https://doi.org/10.1007/s00182-017-0596-4}
  {\path{doi:10.1007/s00182-017-0596-4}}.

\bibitem{DeinekoW13}
Vladimir~G. Deineko and Gerhard~J. Woeginger.
\newblock Two hardness results for core stability in hedonic coalition
  formation games.
\newblock {\em Discret. Appl. Math.}, 161(13-14):1837--1842, 2013.

\bibitem{ElkindFF20}
Edith Elkind, Angelo Fanelli, and Michele Flammini.
\newblock Price of pareto optimality in hedonic games.
\newblock {\em Artif. Intell.}, 288:103357, 2020.

\bibitem{ElkindW09}
Edith Elkind and Michael~J. Wooldridge.
\newblock Hedonic coalition nets.
\newblock In {\em {AAMAS} {(1)}}, pages 417--424. {IFAAMAS}, 2009.

\bibitem{0001MM21}
Angelo Fanelli, Gianpiero Monaco, and Luca Moscardelli.
\newblock Relaxed core stability in fractional hedonic games.
\newblock In Zhi{-}Hua Zhou, editor, {\em Proceedings of the Thirtieth
  International Joint Conference on Artificial Intelligence, {IJCAI} 2021,
  Virtual Event / Montreal, Canada, 19-27 August 2021}, pages 182--188.
  ijcai.org, 2021.
\newblock \href {https://doi.org/10.24963/ijcai.2021/26}
  {\path{doi:10.24963/ijcai.2021/26}}.

\bibitem{FlamminiKMZ21}
Michele Flammini, Bojana Kodric, Gianpiero Monaco, and Qiang Zhang.
\newblock Strategyproof mechanisms for additively separable and fractional
  hedonic games.
\newblock {\em J. Artif. Intell. Res.}, 70:1253--1279, 2021.

\bibitem{GairingS19}
Martin Gairing and Rahul Savani.
\newblock Computing stable outcomes in symmetric additively separable hedonic
  games.
\newblock {\em Math. Oper. Res.}, 44(3):1101--1121, 2019.

\bibitem{GareyJ79}
M.~R. Garey and David~S. Johnson.
\newblock {\em Computers and Intractability: {A} Guide to the Theory of
  {NP}-Completeness}.
\newblock W. H. Freeman, 1979.

\bibitem{HanakaKMO19}
Tesshu Hanaka, Hironori Kiya, Yasuhide Maei, and Hirotaka Ono.
\newblock Computational complexity of hedonic games on sparse graphs.
\newblock In {\em {PRIMA}}, volume 11873 of {\em Lecture Notes in Computer
  Science}, pages 576--584. Springer, 2019.

\bibitem{HarutyunyanLM21}
Ararat Harutyunyan, Michael Lampis, and Nikolaos Melissinos.
\newblock Digraph coloring and distance to acyclicity.
\newblock In {\em {STACS}}, volume 187 of {\em LIPIcs}, pages 41:1--41:15.
  Schloss Dagstuhl - Leibniz-Zentrum f{\"{u}}r Informatik, 2021.

\bibitem{IeongS05}
Samuel Ieong and Yoav Shoham.
\newblock Marginal contribution nets: a compact representation scheme for
  coalitional games.
\newblock In {\em {EC}}, pages 193--202. {ACM}, 2005.

\bibitem{IgarashiE16}
Ayumi Igarashi and Edith Elkind.
\newblock Hedonic games with graph-restricted communication.
\newblock In {\em {AAMAS}}, pages 242--250. {ACM}, 2016.

\bibitem{IgarashiOSY19}
Ayumi Igarashi, Kazunori Ota, Yuko Sakurai, and Makoto Yokoo.
\newblock Robustness against agent failure in hedonic games.
\newblock In Sarit Kraus, editor, {\em Proceedings of the Twenty-Eighth
  International Joint Conference on Artificial Intelligence, {IJCAI} 2019,
  Macao, China, August 10-16, 2019}, pages 364--370. ijcai.org, 2019.
\newblock \href {https://doi.org/10.24963/ijcai.2019/52}
  {\path{doi:10.24963/ijcai.2019/52}}.

\bibitem{ImpagliazzoPZ01}
Russell Impagliazzo, Ramamohan Paturi, and Francis Zane.
\newblock Which problems have strongly exponential complexity?
\newblock {\em J. Comput. Syst. Sci.}, 63(4):512--530, 2001.
\newblock \href {https://doi.org/10.1006/jcss.2001.1774}
  {\path{doi:10.1006/jcss.2001.1774}}.

\bibitem{JansenKMS13}
Klaus Jansen, Stefan Kratsch, D{\'{a}}niel Marx, and Ildik{\'{o}} Schlotter.
\newblock Bin packing with fixed number of bins revisited.
\newblock {\em J. Comput. Syst. Sci.}, 79(1):39--49, 2013.
\newblock \href {https://doi.org/10.1016/j.jcss.2012.04.004}
  {\path{doi:10.1016/j.jcss.2012.04.004}}.

\bibitem{Lampis21}
Michael Lampis.
\newblock Minimum stable cut and treewidth.
\newblock In {\em {ICALP}}, volume 198 of {\em LIPIcs}, pages 92:1--92:16.
  Schloss Dagstuhl - Leibniz-Zentrum f{\"{u}}r Informatik, 2021.

\bibitem{LampisMM18}
Michael Lampis, Stefan Mengel, and Valia Mitsou.
\newblock {QBF} as an alternative to {Courcelle}'s theorem.
\newblock In Olaf Beyersdorff and Christoph~M. Wintersteiger, editors, {\em
  Theory and Applications of Satisfiability Testing - {SAT} 2018 - 21st
  International Conference, {SAT} 2018, Held as Part of the Federated Logic
  Conference, FloC 2018, Oxford, UK, July 9-12, 2018, Proceedings}, volume
  10929 of {\em Lecture Notes in Computer Science}, pages 235--252. Springer,
  2018.
\newblock \href {https://doi.org/10.1007/978-3-319-94144-8\_15}
  {\path{doi:10.1007/978-3-319-94144-8\_15}}.

\bibitem{LampisM17}
Michael Lampis and Valia Mitsou.
\newblock Treewidth with a quantifier alternation revisited.
\newblock In {\em {IPEC}}, volume~89 of {\em LIPIcs}, pages 26:1--26:12.
  Schloss Dagstuhl - Leibniz-Zentrum f{\"{u}}r Informatik, 2017.

\bibitem{LokshtanovMS18}
Daniel Lokshtanov, D{\'{a}}niel Marx, and Saket Saurabh.
\newblock Slightly superexponential parameterized problems.
\newblock {\em {SIAM} J. Comput.}, 47(3):675--702, 2018.

\bibitem{OhtaBISY17}
Kazunori Ohta, Nathana{\"{e}}l Barrot, Anisse Ismaili, Yuko Sakurai, and Makoto
  Yokoo.
\newblock Core stability in hedonic games among friends and enemies: Impact of
  neutrals.
\newblock In Carles Sierra, editor, {\em Proceedings of the Twenty-Sixth
  International Joint Conference on Artificial Intelligence, {IJCAI} 2017,
  Melbourne, Australia, August 19-25, 2017}, pages 359--365. ijcai.org, 2017.
\newblock \href {https://doi.org/10.24963/ijcai.2017/51}
  {\path{doi:10.24963/ijcai.2017/51}}.

\bibitem{Olsen09}
Martin Olsen.
\newblock Nash stability in additively separable hedonic games and community
  structures.
\newblock {\em Theory Comput. Syst.}, 45(4):917--925, 2009.

\bibitem{OlsenBT12}
Martin Olsen, Lars B{\ae}kgaard, and Torben Tambo.
\newblock On non-trivial nash stable partitions in additive hedonic games with
  symmetric 0/1-utilities.
\newblock {\em Inf. Process. Lett.}, 112(23):903--907, 2012.

\bibitem{Peters16a}
Dominik Peters.
\newblock Graphical hedonic games of bounded treewidth.
\newblock In {\em {AAAI}}, pages 586--593. {AAAI} Press, 2016.

\bibitem{Peters17}
Dominik Peters.
\newblock Precise complexity of the core in dichotomous and additive hedonic
  games.
\newblock In {\em {ADT}}, volume 10576 of {\em Lecture Notes in Computer
  Science}, pages 214--227. Springer, 2017.

\bibitem{PetersE15}
Dominik Peters and Edith Elkind.
\newblock Simple causes of complexity in hedonic games.
\newblock In {\em {IJCAI}}, pages 617--623. {AAAI} Press, 2015.

\bibitem{SaadHBDH11}
Walid Saad, Zhu Han, Tamer Basar, M{\'{e}}rouane Debbah, and Are Hj{\o}rungnes.
\newblock Hedonic coalition formation for distributed task allocation among
  wireless agents.
\newblock {\em {IEEE} Trans. Mob. Comput.}, 10(9):1327--1344, 2011.
\newblock \href {https://doi.org/10.1109/TMC.2010.242}
  {\path{doi:10.1109/TMC.2010.242}}.

\bibitem{SliwinskiZ17}
Jakub Sliwinski and Yair Zick.
\newblock Learning hedonic games.
\newblock In Carles Sierra, editor, {\em Proceedings of the Twenty-Sixth
  International Joint Conference on Artificial Intelligence, {IJCAI} 2017,
  Melbourne, Australia, August 19-25, 2017}, pages 2730--2736. ijcai.org, 2017.
\newblock \href {https://doi.org/10.24963/ijcai.2017/380}
  {\path{doi:10.24963/ijcai.2017/380}}.

\bibitem{SungD10}
Shao~Chin Sung and Dinko Dimitrov.
\newblock Computational complexity in additive hedonic games.
\newblock {\em Eur. J. Oper. Res.}, 203(3):635--639, 2010.

\bibitem{Woeginger13}
Gerhard~J. Woeginger.
\newblock A hardness result for core stability in additive hedonic games.
\newblock {\em Math. Soc. Sci.}, 65(2):101--104, 2013.

\end{thebibliography}
